\documentclass{article}

\usepackage[letterpaper, margin=1in]{geometry}

\usepackage{microtype}
\usepackage{graphicx}
\usepackage{booktabs} 
\usepackage[square,numbers,sort]{natbib}
\usepackage{hyperref}

\usepackage{algorithm}
\usepackage[noend]{algorithmic}

\usepackage{amsmath}
\usepackage{amssymb}
\usepackage{mathtools}
\usepackage{amsthm}

\usepackage[capitalize,noabbrev]{cleveref}

\theoremstyle{plain}
\newtheorem{theorem}{Theorem}[section]
\newtheorem{proposition}[theorem]{Proposition}
\newtheorem{lemma}[theorem]{Lemma}
\newtheorem{corollary}[theorem]{Corollary}
\theoremstyle{definition}
\newtheorem{definition}[theorem]{Definition}

\theoremstyle{remark}
\newtheorem{remark}[theorem]{Remark}

\newcount\Includeappendix
\newcommand{\apxref}[1]{%
  \ifnum\Includeappendix=0
    the appendix%
  \else
    Section~\ref{#1}%
  \fi
}

\DeclarePairedDelimiter\abs{\lvert}{\rvert}%
\DeclarePairedDelimiter\norm{\lVert}{\rVert}
\DeclareMathOperator*{\argmax}{arg\,max}

\newcommand{\prob}{\textsc{Rec-APC}}
\newcommand{\E}[1]{\mathbb{E}\left[ #1 \right]}
\newcommand{\ovv}{\overline{V}}
\newcommand{\unv}{\underline{V}}

\usepackage{bm}
\usepackage{kbordermatrix} 

\theoremstyle{definition}
\newtheorem{observation}[theorem]{Observation}
\newtheorem{example}[theorem]{Example}

\usepackage{newfloat}
\usepackage{subcaption}

\newcommand{\IFTHEN}[2]{\STATE \textbf{if} #1 \textbf{then} #2}

\makeatletter
\newcommand{\alglinelabel}[1]{%
  \addtocounter{ALC@line}{-1}
  \refstepcounter{ALC@line}
  \label{#1}
}
\makeatother

\usepackage{enumitem}

\newcommand*{\Scale}[2][4]{\scalebox{#1}{$\displaystyle #2$}}%

\begin{document}

\title{Modeling Churn in Recommender Systems with Aggregated Preferences}

\author{
Gur Keinan%
\thanks{%
    {Technion---Israel Institute of Technology (\url{gur.keinan@campus.technion.ac.il})}}
\and Omer Ben{-}Porat%
\thanks{%
    {Technion---Israel Institute of Technology (\url{omerbp@technion.ac.il})}}
}

\maketitle

\begin{abstract}
While recommender systems (RSs) traditionally rely on extensive individual user data, regulatory and technological shifts necessitate reliance on aggregated user information. This shift significantly impacts the recommendation process, requiring RSs to engage in intensive exploration to identify user preferences. However, this approach risks user churn due to potentially unsatisfactory recommendations. In this paper, we propose a model that addresses the dual challenges of leveraging aggregated user information and mitigating churn risk. Our model assumes that the RS operates with a probabilistic prior over user types and aggregated satisfaction levels for various content types. We demonstrate that optimal policies naturally transition from exploration to exploitation in finite time, develop a branch-and-bound algorithm for computing these policies, and empirically validate its effectiveness.
\end{abstract}

\section{Introduction}

Recommender systems (RSs) have become essential in digital media and e-commerce. They collect massive amounts of data and apply sophisticated techniques to improve user engagement and satisfaction. RSs leverage past user interaction, demographics, and possibly additional information to provide personalized recommendations. For instance, Netflix and Amazon utilize these techniques to offer tailored movie recommendations and product suggestions, respectively. While RSs typically collect and record user data, some have limited access to individual user information or none at all. For instance, regulations such as the General Data Protection Regulation (GDPR) and the California Consumer Privacy Act (CCPA) impose stringent rules on data usage and user consent. Prioritization of privacy is also a trend within commercial companies, e.g., Apple's iOS 14 opt-in device tracking modification that has affected targeted advertising~\cite{kollnig2022goodbye}. In some cases, the RS can record \emph{user sessions} of varying length, but cannot identify users. While these regulatory shifts aim to protect privacy, they pose substantial challenges to RSs in maintaining the same level of service. Due to these privacy-preserving limitations, RSs may at times be limited to utilizing \emph{aggregated user information}, such as clusters of users or personas.

Relying solely on aggregated information significantly affects the recommendation process. When interacting in a user session without accurate individual data, the RS must explore more intensively to understand the current user's preferences. This initial exploration phase is critical for gathering enough data to make accurate recommendations later in the session. However, aggressive exploration bears the risk of \emph{user churn}, where users may leave the system due to receiving unsatisfactory recommendations. For example, when a song-recommendation RS encounters a user with unknown preferences, it could suggest a song from an unconventional genre, which some users enjoy while most users dislike. Although user feedback on that unconventional genre can provide valuable insights into their preferences, it can cause users to leave the RS if the recommendation is off-target; thus, the RS should address the possibility of churn as part of its design.

In this paper, we propose a model to study the intertwined challenges of aggregated user information and churn risk. We call our model $\prob$, standing for \textbf{Rec}ommendation with \textbf{A}ggregated \textbf{P}references under \textbf{C}hurn. Our model assumes that the RS has a probabilistic prior over user types (clusters, personas, etc.) and aggregated satisfaction levels for each content type (genre, etc.) with respect to each user type. Each user session involves an unidentified user whose type is sampled from this prior distribution and is unknown to the RS. The RS recommends content sequentially, receiving binary feedback (like or dislike) from the user. This feedback allows the RS to infer the user type in a Bayesian sense to improve future recommendations. However, careless present recommendations can lead to user churn. The objective of the RS is to maximize user utility, defined as the number of contents the user likes.

\subsection{Our contribution}
Our contribution is three-fold. First, we are among the first to address aggregated user preferences and the risk of churn simultaneously. We propose a model where the general population preferences are known and the type (cluster, persona, etc.) of users are hidden but can be deduced through interaction. This interaction drives user engagement, which we model as the RS's reward. However, due to uncertainty about user type, the RS could generate unsatisfactory recommendations that can cause user churn. Our model poses a novel exploration-exploitation trade-off with the risk of churn, an aspect under-explored in the current literature.

Our second contribution is technical. Analyzing the model shows that optimal policies converge after a finite number of recommendations, symbolizing a transition to pure exploitation. Informally,
\begin{theorem}[Informal statement of Theorem~\ref{thm:convergence}]
For a broad range of instances, the (infinite-length) optimal policy converges.
\end{theorem}

Our third contribution is algorithmic. We leverage our theoretical results to develop a straightforward yet state-of-the-art branch-and-bound algorithm designed explicitly for our setting. As the problem we address can be described as a partially observable Markov decision process (POMDP), we compare our algorithm with the state-of-the-art POMDP benchmark~\cite{kurniawati2009sarsop}. In practice, our algorithm performs better when there is a large variety of user types but is less effective when the number of contents is significantly larger than the number of user types.

\subsection{Related work}
Our work captures several phenomena: First, we assume that the RS has aggregated user information, but no access to individual user information. Second, we have sequential interaction, where we can balance exploration-exploitation trade-offs. And third, our model includes the risk of user churn. Below, we review the relevant literature strands.

\paragraph{Aggregated user information}
Our model follows the trend of using clustered data in the recommendation process~\cite{recommender-systems-for-large-scale-e-commerce-scalable-neighborhood-formation-using-clustering}. In addition to improved efficiency, the use of clustering can increase the diversity and reliability of recommendations \cite{a-clustering-approach-for-personalizing-diversity-in-collaborative-recommender-systems, robust-collaborative-filtering-based-on-multiple-clustering} and handle the sparsity of user preference matrices \cite{recommender-systems-clustering-using-bayesian-non-negative-matrix-factorization}.

Aggregated information also relates to privacy, a topic that has gained much attention recently, following the seminal work of \citet{differential-privacy} on differential privacy. Several works propose RSs that satisfy differential privacy \cite{differentially-private-recommender-systems-building-privacy-into-the-netflix-prize-contenders, differentially-private-collaborative-coupling-learning-for-recommender-systems, differential-privacy-for-collaborative-filtering-recommender-algorithm}. In a broader context, \citet{the-effect-of-online-privacy-information-on-purchasing-behavior-an-experimental-study} have empirically shown that users value their privacy and are willing to pay for it. Several other definitions of privacy were suggested in the literature \cite{an-agent-based-approach-for-privacy-preserving-recommender-systems,
enhancing-privacy-and-preserving-accuracy-of-a-distributed-collaborative-filtering,svd-based-collaborative-filtering-with-privacy}. In our work, we assume that the RS has access to aggregated data, akin to other recent works addressing lookalike clustering~\cite{anonymous-learning-via-look-alike-clustering-a-precise-analysis-of-model-generalization, interactive-and-explainable-point-of-interest-recommendation-using-look-alike-groups}.

The cases where the RS has little information about users or items are called \emph{cold-start} problems. This issue relates to our work because, while we assume access to aggregated information, every user interaction starts from a tabula rasa. Broadly, solutions are divided into data-driven approaches~\cite{a-heterogeneous-information-network-based-cross-domain-insurance-recommendation-system-for-cold-start-users, transfer-meta-framework-for-cross-domain-recommendation-to-cold-start-users, alleviating-data-sparsity-and-cold-start-in-recommender-systems-using-social-behaviour} and method-driven approaches~\cite{personalized-adaptive-meta-learning-for-cold-start-user-preference-prediction, task-adaptive-neural-process-for-user-cold-start-recommendation, meta-matrix-factorization-for-federated-rating-predictions} (see \citet{user-cold-start-problem-in-recommendation-systems-a-systematic-review} for a recent survey). 

This paper is \emph{inspired} by clustering, privacy, and cold-start problems. However, our model only assumes access to aggregated information and abstracts the reasons why individual information is unavailable. Notably, we do not propose techniques to cluster users, address privacy concerns, or provide new approaches for the cold-start problem.

\paragraph{Sequential recommendation with churn}
Our model falls under Markov Decision Process (MDP) modeling for RSs~\cite{an-mdp-based-recommender-system}, where the belief over user types represents the state. Alternatively, we can formalize it as a Partially Observable MDP (POMDP), where the state corresponds to the user type that remains constant but is initially unknown. Both MDP and POMDP modeling are well-studied in the literature of RSs~\cite{optimal-recommendation-to-users-that-react-online-learning-for-a-class-of-pomdps, interactive-recommendation-with-user-specific-deep-reinforcement-learning, recommendation-as-a-stochastic-sequential-decision-problem}, and typically the main task is to learn the underlying model~\cite{empirical-evaluation-of-gated-recurrent-neural-networks-on-sequence-modeling, usage-based-web-recommendations-a-reinforcement-learning-approach}.

We model user churn as part of the sequential recommendation process, thereby generating an exploration-exploitation trade-off. User churn is an integral part of many systems, and most of the literature addresses user churn prediction~\cite{user-retention-a-causal-approach-with-triple-task-modeling,quantifying-and-leveraging-user-fatigue-for-interventions-in-recommender-systems} and techniques to retain users~\cite{e-government-deep-recommendation-system-based-on-user-churn,surrogate-for-long-term-user-experience-in-recommender-systems}. Several recent works~\cite {maximizing-cumulative-user-engagement-in-sequential-recommendation-an-online-optimization-perspective, returning-is-believing-optimizing-long-term-user-engagement-in-recommender-systems, partially-observable-markov-decision-process-for-recommender-systems, cao2020fatigue} adopt a similar approach to ours, directly modeling user churn due to irrelevant recommendations. 


The paper most relevant to ours is the one by \citet{modeling-attrition-in-recommender-systems-with-departing-bandits}. They focus on the problem of online learning with the risk of user churn under user uncertainty. While they study the stochastic variant, where the user-type preferences are initially unknown, their analysis is restricted to one user and multiple categories, or two users and two categories. In contrast, we assume complete information, but we study general matrices of any size.

\section{Problem Definition}
\label{sec:problem-definition}

In this section, we formally introduce our model alongside the optimization problem the recommender system aims to solve. We then provide an illustrative example and explain why this problem necessitates careful planning by showing the sub-optimality of naive and seemingly optimal solutions.

\subsection{Our Model}
In this section, we formally define the \textbf{Rec}ommendation with \textbf{A}ggregated \textbf{P}references under \textbf{C}hurn, or $\prob$ for abbreviation. An instance of the problem consists of several elements, as we formally describe below.

We consider a set $U$ of users who interact with the RS.
The RS does not have access to user information; instead, it relies on aggregated information about the users. Specifically, we assume that there is a finite set $M$ of \emph{user types}, with each type representing a persona, i.e., a cluster of homogeneous users with similar preferences. Each user in $U$ is associated with precisely one user type in $M$ according to its preferences. Accordingly, we denote by $m(u)\in M$ the type of user $u\in U$. The RS has a prior distribution $\bm{q}$ on the elements of $M$, that is $\bm{q} \in \Delta(M)$. This prior distribution reflects the likelihood of a new user being of any given type. 

In this paper, we assume that contents are abstracted into broader \emph{categories}, each representing a group of similar items. This abstraction allows the RS to recommend content by selecting from a finite set $K$ of categories.

Lastly, there is a user-type preferences matrix $\bm{P}$ with dimensions $[0,1]^{K \times M}$. Each element $\bm{P}(i,j)$ signifies the likelihood that category $i$ satisfies a user of type $j$. We stress that $\bm{P}$ contains information on user types but not on specific users. The RS has complete information on the user-type preference matrix $\bm P$ and prior user-type distribution $\bm q$.

\paragraph{User session}
A user session starts when a user $u \in U$ enters the RS. The RS lacks access to any information about $u$; thus, it only knows that $m(u)$ is distributed according to $\bm q$. The session consists of rounds. In each round $t$, for $t\in \{1,\dots,\infty\}$, the RS recommends a category $k_t \in K$. Afterward, the user gives binary feedback: They either like the item, occurring with a probability of $P(k_t, m(u))$, or they dislike it with the complementary probability. If the user likes the recommended category, the RS receives a reward of $1$, and the session continues for another round. However, when the user ultimately dislikes a recommended category, the RS earns a reward of $0$, and the session concludes as the user leaves the RS.

\paragraph{Recommendation policy}
The RS produces recommendations according to a recommendation \emph{policy}. Recommendation policies can depend solely on the current user session and the history within the session. That is, in round $t$ of the session, the policy can depend on histories of the form $\left(K,\{0,1\} \right)^{t-1}$, where each tuple comprises a recommended category and its corresponding binary feedback. However, since a feedback of zero (dislike) leads to ending the session, we can succinctly represent a policy as a sequence of recommendations. Namely, we represent a policy $\pi$  as the infinite series of categories $\pi = (k_t)_{t=1}^{\infty}$. 
Observe that $k_t$ will only be recommended if the user liked $k_{t-1}$, as otherwise the user would leave. For convenience, we use $\pi[t:]$ to denote the recommendation sequence of $\pi$ starting from the $t$'th round, that is,  $\pi[t:] = \left( k_{i} \right)_{i=t}^{\infty}$.

\paragraph{Social welfare}
We describe the utility of each user type as the count of times it engages with the platform and provides positive feedback, i.e., the number of likes. Let $F(m(u),\pi)$ represent the r.v. that counts the number of likes given by a user of type $m(u)$. This count is influenced by both the user type $m(u)$ and the recommendation policy $\pi$. We define $V^\pi$ as the expected social welfare, which is the mean utility of users under the recommendation policy $\pi$. That is,
\begin{equation*}
    \Scale[0.97
    ]{V^\pi = \E{\frac{1}{\abs{U}}\sum_{u \in U} F(m(u),\pi)} = \sum_{m \in M}\bm q(m)\E{F(m,\pi)}},
\end{equation*}
where the equality follows from the definition of the prior $\bm q$. To be consistent with the literature, we refer to $V^\pi$ as the \emph{value function}; we use both terms interchangeably.
As the goal of the recommender is to maximize expected social welfare, we define \emph{optimal policy} $\pi^\star$ as $\pi^\star \in \argmax_{\pi} V^\pi$ and the corresponding optimal expected social welfare $V^\star$ as $V^\star = V^{\pi^\star}$.

\paragraph{Useful notation}
We emphasize that a policy $\pi$ is a function of $\bm {P}$ and $\bm {q}$, expressed as $\pi=\pi(\bm P, \bm q)$. For abbreviation, it is also convenient to denote this as $\pi(\bm q)$ when $\bm P$ is known from the context. This notation extends to any belief $\bm b \in \Delta(M)$, not just the prior $\bm q$. We make the same abuse of notation for the value function $V^\pi$. That is, we let $V^\pi(\bm b)=\sum_{m \in M}\bm b(m) \E{F(m,\pi)}$ and stress that $V^\pi(\bm q)=V^\pi$. Additionally, we use the star notation $\pi^\star(\bm b), V^\star(\bm b)$ to denote optimal policy and value function w.r.t. the belief $\bm b$. Finally, we let $p_{k}(\bm{b})$ denote the expected \emph{immediate reward}, specifically, the probability that a user will like category $k$ assuming that their type is drawn from $\bm{b}$; that is, $p_{k}(\bm{b}) = \sum_{m \in M} \bm{b}(m) \bm{P}(k, m)$.

\subsection{Bayesian Updates}
At the beginning of each user session, the RS is only informed about the prior $\bm q$. However, it rapidly acquires more information about the user through their feedback. For instance, if a user likes a recommendation of an exotic category favored by only a small subset of user types, we can conclude that the user probably belongs to that subset of types. The RS can then update its \emph{belief} over the current user type, where the belief is a point in the user type simplex, and use it in its future suggestions. 

We employ Bayesian updates to incorporate the new information after user feedback. Starting from a belief $\bm b$ and recommending a category $k$, we let $\tau(\bm{b}, k)$ denote the new belief over user types in case of positive feedback. Namely, $\tau$ is a function $\tau: \Delta(M) \times K \to \Delta(M)$ such that for every belief $\bm{b} \in \Delta(M)$, category $k \in K$ and type $m \in M$,
\begin{equation}\label{eq:bayesian update}
    \tau(\bm{b}, k)(m) =  \frac{\bm{b}(m) \cdot \bm{P}(k,m)}{\sum_{m' \in M} \bm{b}(m') \cdot \bm{P}(k,m')}.
\end{equation}

We stress that Bayesian updates are crucial only after positive feedback, as the RS can utilize this new information for future recommendations. Negative feedback, though still informative, ends the session, preventing the RS from using the new information. We can further use Bayesian updates to obtain a recursive definition of the value function.
\begin{observation}
    \label{obs:recursive-formula-of-the-value-function}
    For every policy $\pi$ and belief $\bm{b} \in \Delta(M)$,
    \[
        V^{\pi}(\bm{b}) = p_{\pi_1}(\bm{b}) \left( 1 + V^{\pi[2:]}(\tau(\bm{b}, \pi_1)) \right).
    \]
\end{observation}
Observation~\ref{obs:recursive-formula-of-the-value-function} provides a Bellman equation-like representation of the value function by isolating the immediate reward of the first round $p_{\pi_1}(\bm{b})$ and the future rewards (implicitly discounted by $p_{\pi_1}(\bm{b})$). It showcases the fundamental exploration-exploitation in our setting: On the one hand, exploitation involves selecting the category $k$ that currently maximizes $p_{k}(\bm{b})$, focusing on immediate reward. On the other hand, exploration aims at steering the updated belief $\tau(\bm{b}, k)$ to a more informative position, thus increasing future rewards. Next, we use the notion of belief walk to assess how policies navigate the belief simplex.

\begin{definition}[Belief Walk]
    The \emph{belief walk} induced by a policy $\pi$ starting at belief $\bm{b}$ is the sequence $( \bm{b}^{\pi, \bm{b}}_t)_{t=1}^{\infty}$, where $\bm{b}^{\pi, \bm{b}}_1 = \bm{b}$ and for all $t > 1$ , $\bm{b}^{\pi, \bm{b}}_t = \tau(\bm{b}^{\pi, \bm{b}}_{t-1}, \pi_{t-1}).$
\end{definition}

To enhance clarity, we provide geometric illustrations of belief walks induced by various recommendation policies in \apxref{sec:belief walks}. Using this notion, we can delineate the value function with closed-form expressions, which will be useful in later analyses. 
\begin{lemma}\label{lemma:closed-form-representations-of-the-value-function}
    For every policy $\pi$ and belief $\bm{b} \in \Delta(M)$,
    {
    \thickmuskip=2mu plus 2mu
    \[
        V^\pi(\bm{b}) = \sum_{t=1}^{\infty} \prod_{j=1}^{t} p_{\pi_j}(\bm{b}^{\pi, \bm{b}}_j) = \sum_{m \in M} \bm{b}(m) \cdot \sum_{t=1}^{\infty} \prod_{j=1}^{t} \bm{P}(\pi_j, m).
    \]
    }%
\end{lemma}
We defer the proof of Lemma~\ref{lemma:closed-form-representations-of-the-value-function} and all other proofs to the appendix.

\subsection{Illustrating Example and Sub-optimality of Myopic Recommendations}
\begin{example}\label{example:body}
    \normalfont
    Consider a $\prob$ instance with two user types $M=\{m_1, m_2\}$ and two categories $K=\{k_1, k_2\}$. The preference matrix and the user type prior are:
    \renewcommand{\kbldelim}{(}
    \renewcommand{\kbrdelim}{)}
    \[
        \bm{P} = \kbordermatrix{
            & m_1 & m_2  \\
            k_1 & 0.95 & 0.1  \\
            k_2 & 0.79 & 0.81
        }, \quad
        \bm{q} =
        \kbordermatrix{
            &  \\
            m_1 & 0.5 \\
            m_2 & 0.5
        }.
    \]

    To interpret these values, note that user type $m_2$ likes category $k_1$ with a probability of $0.1$. In addition, recommending category $k_2$ to a random user yields an expected immediate reward of $p_{k_2}(\bm q)=0.5 \cdot 0.79 + 0.5 \cdot 0.81 = 0.8$.


    A prudent policy to adopt is the \emph{myopic} policy $\pi^m$, which is defined to be the policy that recommends the highest yielding category in each round, and afterward updates the belief. In this instance, the myopic policy is $\pi^m = \left( k_2 \right)_{t=1}^{\infty}$, as $k_2$ provides the better expected immediate reward compared to $k_1$ for the prior, and the belief update will only increase the preference for $k_2$.
    Under this policy, the expected social welfare is the sum of two infinite geometric series, one for each user type; namely, $V^{\pi^m} = 0.5\cdot \frac{0.79}{1 - 0.79} +0.5 \cdot \frac{0.81}{1 - 0.81} =4.01$.

    While unintuitive at first glance, the optimal policy $\pi^\star = \left( k_1 \right)_{i=1}^{\infty}$ achieves an expected value of $V^\star = 0.5\cdot\frac{0.95}{1-0.95} + 0.5\cdot\frac{0.1}{1-0.1} = 9.55$, outperforming the myopic policy.
\end{example}
Beyond exemplifying our setting and notation, Example~\ref{example:body} demonstrates that myopic policies can be sub-optimal. In fact, the next proposition shows that the sub-optimality gap could be arbitrarily large, highlighting the need for adequate planning.
\begin{proposition} \label{prop:myopic-policy-suboptimality}
    For every $d \in \mathbb R_+$, there exists an instance \( \left\langle M, K, \bm{q}, \bm{P} \right\rangle \) such that
    \( V^{\pi^{m}} \cdot d \leq V^{\star} \),
    where $\pi^{m}$ is the myopic policy for the instance.
\end{proposition}

\section{Dynamic Programming-based Algorithms}
\label{sec:approximating-the-optimal-policy}

In this section, we present two dynamic programming-based approaches that provide approximations of the optimal expected social welfare. Although inefficient in the general case, such solutions are useful for rectangular instances, namely scenarios in which the number of types $|M|$ or categories $|K|$ is small. We first define the notion of approximation through a finite horizon and then present several methods for achieving it.

While our model lacks an explicit discount factor, Observation~\ref{obs:recursive-formula-of-the-value-function} indicates that an implicit discount emerges through $p_{k}(\bm{b})$. Thus, similar to well-known results in MDPs, one can approximate the value function over an infinite horizon by addressing a finite-horizon problem with a sufficiently large horizon~\citep[Sec.~17.2]{ai-a-modern-approach}. To that end, we define the finite-horizon value function with a horizon $H$ as $V_H^{\pi}(\bm{b})=\sum_{m \in M}\bm q(m)\E{\min \{H, F(m,\pi)\}}$. In other words, $V_H^{\pi}(\bm{b})$ is the value function given that the session terminates after $H$ rounds. Lemma~\ref{lemma:closed-form-representations-of-the-value-function} suggests that $V_H^{\pi}(\bm{b})$ is also given by $V_H^{\pi}(\bm{b}) = \sum_{t=1}^{H} \prod_{j=1}^{t} p_{\pi_{j}}(\bm{b}^{\pi, \bm{b}}_j)$.

Next, we denote $p_{\max}$ as the largest entry in the matrix~$\bm P$, assuming $p_{\max} < 1$ \footnote{If \( p_{\max} = 1 \), adopting a policy \( \pi = (k)_{t=1}^\infty \) where \( \bm{P}_{k,m} = 1 \) for some type \( m \) results in infinite expected social welfare.}, and use it to bound the gap between infinite and finite horizon optimal policies.
\begin{lemma} \label{lemma:finite-horizon-approximation}
 For any $\varepsilon >0$, it holds that $V^{\star} \leq \max_{\pi'} V_{H(\varepsilon)}^{\pi'}+ \varepsilon$, where 
 $H(\varepsilon) = \left\lceil \log_{p_{\max}} \frac{\varepsilon (1 - p_{\max})}{p_{\max}} \right\rceil$.
\end{lemma}
Lemma~\ref{lemma:finite-horizon-approximation} reduces the task of finding an approximately optimal policy to finding an (exact or approximate) optimal policy for the finite-horizon case. For the rest of the section, we develop such solutions for a horizon of $H$.
\paragraph{Small number of categories}
Consider the brute-force approach, which evaluates all $K^H$ possible policies to find the optimal one. This becomes infeasible for large $H$, even with small $K$. However, we can exploit the order-invariance of belief updates: the belief depends only on the multiset of selected categories, not their sequence. Thus, redundant calculations can be avoided by iteratively grouping sequences with the same multiset and retaining only the highest-performing sequence from each group.
\begin{proposition} \label{prop:backward-induction-approximation}
  We can find $\pi^\star_H = \argmax_{\pi'} V^{\pi'}_H$ in a runtime of $O\left( (H + \abs{K}) ^ {\abs{K}} \cdot \abs{K} \cdot \abs{M} \right)$.
\end{proposition}
\paragraph{Small number of user types}
If $M$ is small, we can employ a different approach. Recall that beliefs are points in the simplex $\Delta(M)$. Thus, we can discretize the belief simplex and execute dynamic programming. Specifically, we adopt the approach of \citet{point-based-value-iteration-an-anytime-algorithm-for-pomdps} and obtain:
\begin{proposition} \label{prop:discretization-approximation}
  We can find a policy $\pi$ that satisfies $V^{\pi}_H \geq \max_{\pi'} V^{\pi'}_H - \varepsilon$ in a runtime of \\ $O\left( H \cdot \abs{K} \cdot \abs{M} \cdot \left( \frac{1}{\varepsilon \cdot (1 - p_{\max})^2} \right)^{2\abs{M} - 2} \right)$.
\end{proposition}

\section{Convergence of Optimal Policies}
\label{sec:convergence-of-the-optimal-policy}
In Section~\ref{sec:approximating-the-optimal-policy}, we showed that finite horizon analysis suffices to approximate optimal policies. Here, we establish an even more fundamental property for a broad family of instances: After a finite number of rounds, the optimal policy becomes fixed, transitioning from exploration to pure exploitation.

We focus our attention on instances where the preference matrix $\bm{P}$ has \emph{distinct preferences} (DP), meaning all entries of $\bm{P}$ are unique. This ensures that every observation provides unambiguous information about user types, enabling informative observations. A discussion of how this theorem generalizes to other instances is provided in \apxref{sec:convergence-of-the-optimal-policy-proofs}. Our main result is as follows:

\begin{theorem}\label{thm:convergence}
   Fix any DP-$\prob$ instance. There exists $T < \infty$ such that for any $t \geq T$, it holds that $\pi^\star_{t+1}=\pi^\star_{t}$.
\end{theorem}

Theorem~\ref{thm:convergence} introduces a natural complexity measure through the time until convergence. Through the proof's construction, we can identify both when convergence occurs and determine the optimal policy from that point onward. Hence, instances with faster convergence require less computational effort as fewer possibilities need to be explored before identifying the final repeating recommendation.

While Theorem~\ref{thm:convergence} only establishes the convergence, our empirical analysis in Section~\ref{sec:experiments} reveals that this transition typically occurs relatively fast in practice.

\begin{proof}[\normalfont\bfseries Proof Sketch of Theorem~\ref{thm:convergence}]
We begin by introducing an instance-dependent parameter $c$ that captures crucial structural 
relationships within the user preference matrix. While the exact definition of $c$ is deferred 
to \apxref{sec:convergence-of-the-optimal-policy-proofs}, intuitively, $c$ encodes values like the degree of heterogeneity among user types---quantified by the minimal separation between entries in each 
row of $\bm{P}$---and the minimum probability in the user-type distribution $\bm{q}$. Crucially, 
$c > 0$ holds for all UP-$\prob$ instances.


Given a user type $m \in M$ and $\delta > 0$, we say a belief $\bm{b}$ is \emph{$(\delta,m)$-concentrated} if $\bm{b}(m) \geq 1-\delta$. Typically, the prior $\bm{q}$ is \emph{$c$-unconcentrated}, meaning $\bm{q}(m) < 1-c$ for all $m \in M$, with probability mass distributed across multiple types.

Our first result demonstrates that for unconcentrated beliefs, the value function strictly increases by a significant gap during consecutive steps of the optimal policy.

\begin{theorem}\label{thm:gap-between-value-function}
For any $\delta$-unconcentrated belief $\bm{b} \in \Delta(M)$ it holds that $V^{\star}(\tau(\bm{b}, \pi^{\star}_1(\bm{b}))) - V^{\star}(\bm{b}) \geq \frac{\delta \cdot (1 - \delta) \cdot c^2}{1 - c}$.
\end{theorem}
Lemma~\ref{lemma:closed-form-representations-of-the-value-function} implies that the value function is bounded above by $\frac{1-c}{c}$, since
\begin{equation}\label{eq:bounded-value}
V^\pi(\bm{b}) = \sum_{t=1}^{\infty} \prod_{j=1}^{t} p_{\pi_j}(\bm{b}^{\pi, \bm{b}}_j)    \leq \frac{p_{\max}}{1-p_{\max}} \leq \frac{1-c}{c}.
\end{equation}
Using potential arguments, we derive the following.
\begin{corollary}\label{corr:limited-unconcetrated}
For a fixed $\delta > 0$, the optimal belief walk initiating from prior $\bm{q}$ can contain at most $H = \left\lceil\frac{(1 - c)^2}{\delta \cdot (1 - \delta) \cdot c^3}\right\rceil$ $\delta$-unconcentrated beliefs.
\end{corollary}
Corollary~\ref{corr:limited-unconcetrated} ensures that the belief walk ultimately enters a sub-space of concentrated beliefs. The next theorem suggests that the optimal recommendation in concentrated beliefs is myopic; therefore, as long as the belief walk remains concentrated, the optimal policy is fixed.
\begin{theorem}\label{thm:myopic-near-boundary}
For every user type $m \in M$ and $(\frac{c^2}{4}, m)$-concentrated belief $\bm{b}$, $\pi^{\star}_1(\bm{b}) = \argmax_{k \in K} \bm{P}(k, m)$.
\end{theorem}
However, we still have one edge case to cover: The belief walk could potentially transition from a $(\frac{c^2}{4},m)$-concentrated belief for some arbitrary $m \in M$ to another $(\frac{c^2}{4},m')$-concentrated belief for $m' \neq m$. If such a case occurs, the optimal policy will diverge. To that end, we provide the following lemma.
\begin{lemma}\label{lemma:concentrated-transition}
For any two distinct user types $m, m' \in M$ and $(\frac{c^2}{4},m)$-concentrated belief $\bm{b}$, $\tau(\bm b, \pi^\star_1(\bm b))$ is not $(\frac{c^2}{4},m')$-concentrated.
\end{lemma}
Lemma~\ref{lemma:concentrated-transition} guarantees that transitions from concentrated beliefs of one type to another must include rounds with unconcentrated beliefs; however, Corollary~\ref{corr:limited-unconcetrated} limits the number of such transitions. Together, Corollary~\ref{corr:limited-unconcetrated} and Lemma~\ref{lemma:concentrated-transition} indicate that after some finite time $T$ the belief remains $(\frac{c^2}{4},m')$-concentrated for one single type $m$ indefinitely, and Theorem~\ref{thm:myopic-near-boundary} suggests that the optimal policy converges to the myopic, fixed policy from $T$ onward.
\end{proof}

\section{Branch-and-Bound Algorithm}
\label{sec:branch-and-bound-algorithm}
In this section, we introduce a branch-and-bound (B\&B) algorithm tailored to our setting, outlined in Algorithm~\ref{bb-algorithm}. The B\&B approach is widely used in sequential decision-making and combinatorial optimization~\cite{learning-to-branch-with-tree-mdps,reinforcement-learning-for-branch-and-bound-optimisation-using-retrospective-trajectories,learning-to-search-in-branch-and-bound-algorithms}. Its effectiveness hinges on the quality of the bounds used to evaluate the search space and eliminate suboptimal branches. For our problem, we derive these bounds based on the recursive structure of the value function $V^\pi$. Specifically, for any policy $\pi$ and positive integer $h \in \mathbb{N}$, the value function $V^\pi$ can be expressed as:
\begin{equation}
  \label{eq:value-function-with-prefix-considerations}
  V^\pi = \sum_{t=1}^{h} \prod_{i=1}^{t} p_{\pi_i}(\bm{b}^{\pi, \bm{q}}_i) + \prod_{i=1}^{h} p_{\pi_i}(\bm{b}^{\pi, \bm{q}}_{i}) \cdot V^{\pi[h+1:]}(\bm{b}^{\pi, \bm{q}}_{h+1}).
\end{equation}

Equation~\eqref{eq:value-function-with-prefix-considerations} decomposes $V^\pi$ into two components: the cumulative rewards for the first $h$ rounds and a recursive term representing the discounted expected value of future rounds. Crucially, substituting $V^{\pi[h+1:]}(\bm{b}^{\pi, \bm{q}}_{h+1})$ with an upper or lower bound yields corresponding bounds for $V^\pi$.

Building on this, we now propose an upper bound. Intuitively, imagine that the RS is entirely certain about the user type. That is, the type would still be sampled according to the belief $\bm b$, but the RS could pick a policy conditioned on the sampled type. In such a case, the RS would pick the type's favorite category indefinitely, leading to an expected social welfare of:
\[
  V^U(\bm{b}) = \sum_{m \in M} \bm{b}(m) \cdot \max_{k \in K} \frac{\bm{P}(k, m)}{1 - \bm{P}(k, m)}.
\]
We stress that the above upper bound is not necessarily attainable. As for the lower bound, we compute the value of the best fixed-action policy w.r.t. the belief, namely,
\[
  V^L(\bm{b}) = \max_{k \in K} \sum_{m \in M} \bm{b}(m) \cdot \frac{\bm{P}(k, m)}{1 - \bm{P}(k, m)}.
\]
Note that the lower bound \emph{is attainable} as it is the value of a valid policy (the best fixed-action policy). 
\begin{lemma} \label{lemma:upper-lower-bound}
  For any belief $\bm b \in \Delta(M)$, it holds that $    V^L(\bm b) \leq V^\star(\bm b) \leq V^U(\bm b)$.
  
\end{lemma}
For ease of notation, for any prefix $\Pi \in \bigcup_{h=1}^{\infty} K^h$ we denote by $\ovv_\Pi, \unv_\Pi$ the integration of $V^U, V^L$ into Equation~\eqref{eq:value-function-with-prefix-considerations}, respectively. Using this notation, $\ovv_\Pi$ acts as an upper bound for the value of any policy that starts with the prefix $\Pi$. Additionally, $\unv_\Pi$ serves as a lower bound of the maximal value of all policies that begin with $\Pi$; namely, $\unv_\Pi \leq \max_{\pi:\text{ begins with }\Pi} V^\pi$. 
\begin{algorithm}[t]
  \caption{B\&B Algorithm for $\prob$}
  \label{bb-algorithm}
  \begin{algorithmic}[1]
    \REQUIRE Instance $\langle M, K, \bm{q}, \bm{P} \rangle$, error term $\varepsilon > 0$
    \ENSURE $\varepsilon$-approximate policy and its value
    \STATE $\tilde \Pi \gets \varnothing$ \COMMENT{The empty prefix} \label{bnbalg:empty_prefix}
    \STATE $\tilde V \gets V^{L}(\bm{q})$ \COMMENT{Lower bound of the empty prefix}
    \STATE $L \gets \text{empty queue}$
    \STATE Append $\tilde \Pi$ to $L$
    \WHILE{$L \neq \emptyset$} 
    \STATE Pop a prefix $\Pi$ from $L$
    \IFTHEN{$\tilde V < \unv_{\Pi}$}{$\tilde V \gets \unv_{\Pi}, \tilde \Pi \gets \Pi$ \alglinelabel{bnbalg:refine}} 
    \FOR{$k \in K$} 
    \IFTHEN{$\ovv_{\Pi \oplus k}-\tilde V > \varepsilon$}{Append $\Pi \oplus k$ to $L$} \alglinelabel{bnbalg:branching}
    \ENDFOR
    \ENDWHILE
    \STATE \textbf{Return} $\tilde \Pi$, $\tilde V$
  \end{algorithmic}
\end{algorithm}

We are ready to present Algorithm~\ref{bb-algorithm}. The algorithm takes as input an instance and an error term $\varepsilon$, outputting an $\varepsilon$-approximately optimal policy and its value. It maintains two key variables: the current best prefix $\tilde \Pi$ and its corresponding value $\tilde V$. Using a queue $L$ to systematically explore policy prefixes, the algorithm implements two critical operations. In Line~\ref{bnbalg:refine}, it performs value refinement by comparing $\tilde V$ against the lower bound $\unv_{\Pi}$ of the examined prefix $\Pi$, updating both $\tilde V$ and $\tilde \Pi$ when an improvement is found. Then, in Line~\ref{bnbalg:branching}, it considers all possible one-step branching of $\Pi$ by appending a category $k$. For each extended prefix $\Pi \oplus k$, it calculates its upper bound $\ovv_{\Pi \oplus k}$. If the potential improvement $\ovv_{\Pi \oplus k} - \tilde{V}$ is more significant than $\varepsilon$, the extended prefix is added to the queue for further exploration. Otherwise, the branch is pruned as it cannot lead to a better solution within the desired accuracy.

The next theorem ensures the $\varepsilon$-optimality of Algorithm~\ref{bb-algorithm}.
\begin{theorem}\label{thm:bb-algorithm-bounded-error}
  For any input $\langle M, K, \bm{q}, \bm{P} \rangle$, $\varepsilon$, Algorithm~\ref{bb-algorithm} terminates after a finite number of steps and returns a value $\tilde V$ such that $V^{\star} - \tilde V \leq \varepsilon$, and a prefix $\tilde \Pi$ such that $\unv_{\tilde \Pi} = \tilde V$.
\end{theorem}
\begin{remark}
  \normalfont
  The output prefix $\tilde{\Pi}$ is a finite sequence of categories, while a policy is defined as an infinite sequence. We can extend $\tilde{\Pi}$ to an approximately optimal policy by using $\tilde{\Pi}$ for the first $h = |\tilde{\Pi}|$ rounds, and then repetitively picking the best-fixed category with respect to $\bm{b}^{\Pi, \bm{q}}_{h+1}$.
\end{remark}


\section{Experiments}\label{sec:experiments}

In this section, we conduct experiments with two goals in mind. First, we complement our convergence result from Theorem~\ref{thm:convergence} by demonstrating that, in practice, the belief walk converges rather quickly. Second, using simulated data, we compare the performance of Algorithm~\ref{bb-algorithm} to a state-of-the-art benchmark. 

\paragraph{Simulation details}
We generate instances using a random sampling procedure. We generate $\bm P$ by independently sampling latent vectors for a given number of categories and user types. Namely, we sample a latent vector for each user type and category from a normal distribution, computing entries of $\bm P$ as negated cosine distances (representing a user's affinities to a category), and normalizing these entries. We generate $\bm q$ by independently sampling logits from a normal distribution. Then, we transform them into a categorical distribution through the softmax function. 
The simulations were conducted on a standard CPU-based PC. Further details appear in \apxref{sec:auxiliary-details-about-the-experiments}. 

\begin{figure}
   \centering
   \includegraphics[width=.5\linewidth]{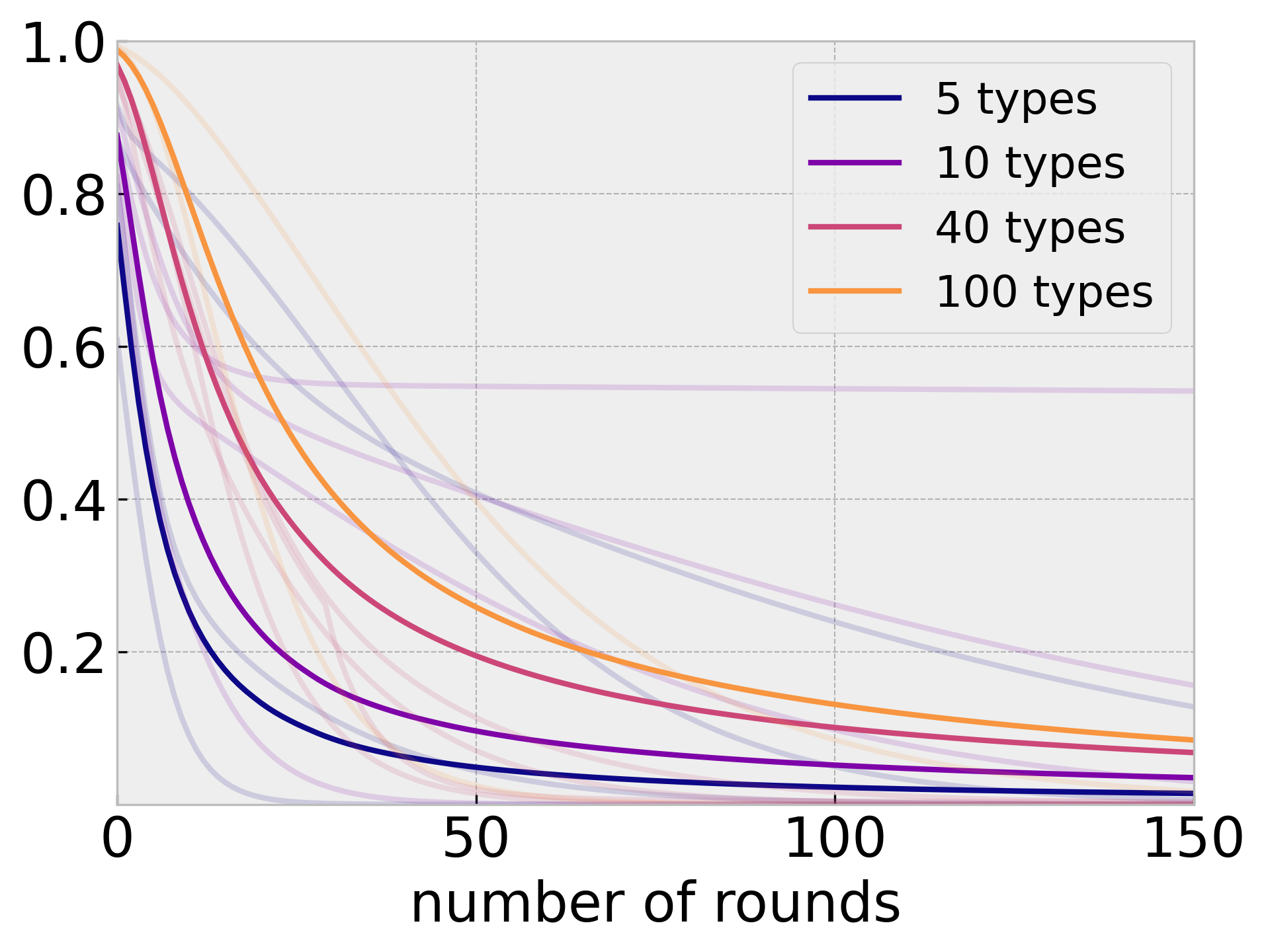}
   \caption{Convergence of beliefs under optimal policies. The number of categories is set to $10$. The x-axis is the number of rounds, and the y-axis is uncertainty in the user type. The transparent lines illustrate a few individual runs and the solid lines are averages over $500$ runs.}
   \label{fig:uncertainty}
\end{figure}

\paragraph{Convergence of beliefs}
Theorem \ref{thm:convergence} establishes that the optimal policy eventually converges to picking a fixed category. While the theorem guarantees convergence, it does not provide explicit rates. Equivalently, convergence can be analyzed in terms of certainty about a user's type, represented by proximity to the vertex to which the belief walk converges (recall the proof sketch of Theorem~\ref{thm:convergence}). Since beliefs update according to Bayes' rule, they converge at a geometric rate once the policy becomes fixed. In other words, further exploration yields diminishing returns when a belief is sufficiently close to a vertex. Thus, it is tempting to assume that the optimal policy myopically maximizes value for that vertex. On the other hand, a poorly chosen myopic policy can fail drastically, as Proposition~\ref{prop:myopic-policy-suboptimality} illustrates. We resolve these conflicting observations through simulations.

Figure \ref{fig:uncertainty} shows how \emph{uncertainty} in user type, defined as the $l_1$-distance from the vertex to which the belief converges under the optimal policy, evolves throughout the session. We vary the number of user types while fixing the number of categories. For each problem size, we report the averaged uncertainty and several individual runs. Despite the heterogeneity of individual runs, their geometric convergence property roughly transfers to averaged curves: Exponential functions fitted to these curves are almost identical to the originals, with correlation coefficients of at least $R^2=0.98$. This matches our intuition that early rounds are most important in terms of both expected reward and information.

Analyzing individual runs reveals notable patterns. While in some sessions, the optimal policies are fixed from the start, in others, recommendations switch (as characterized by jumps in the slope). This reflects the short-term vs. long-term reward trade-off discussed throughout the paper: The optimal policy may initially prioritize immediate rewards before switching to a riskier recommendation that increases certainty and earns more in the long run. 
Despite this, all the presented curves strictly decrease, suggesting that certainty increases monotonically. However, we found that in rare cases, the optimal policy can move away from a vertex before converging to it.
This resolves the above conflict: Even if the belief approaches some vertex, the optimal policy may eventually lead to a different vertex. We exemplify this behavior in \apxref{sec:belief walks}.

\begin{figure*}
\centering
\begin{subfigure}{0.02\textwidth}
   \centering
   \includegraphics[width=\linewidth]{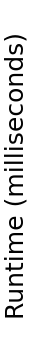}
   \vspace{0.7cm} 
\end{subfigure}%
\hspace{0.1cm} 
\begin{subfigure}{0.3\textwidth}
   \centering
   \includegraphics[width=\linewidth]{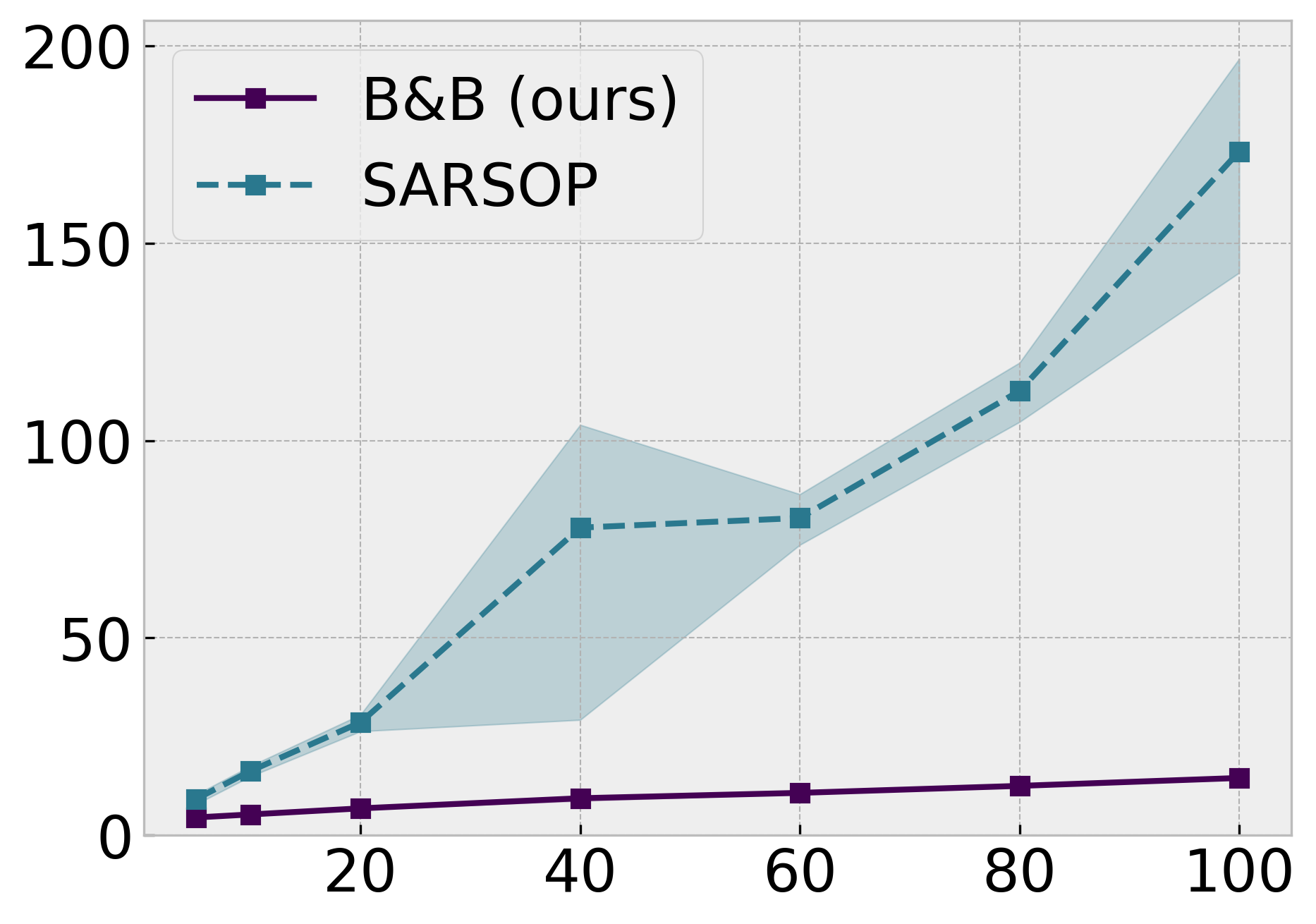}
   \caption{x-axis varies types; 10 categories}
   \label{fig:types_time}
   \vspace{0.5cm} 
\end{subfigure}%
\hspace{0.2cm} 
\begin{subfigure}{0.307\textwidth}
   \centering
   \includegraphics[width=\linewidth]{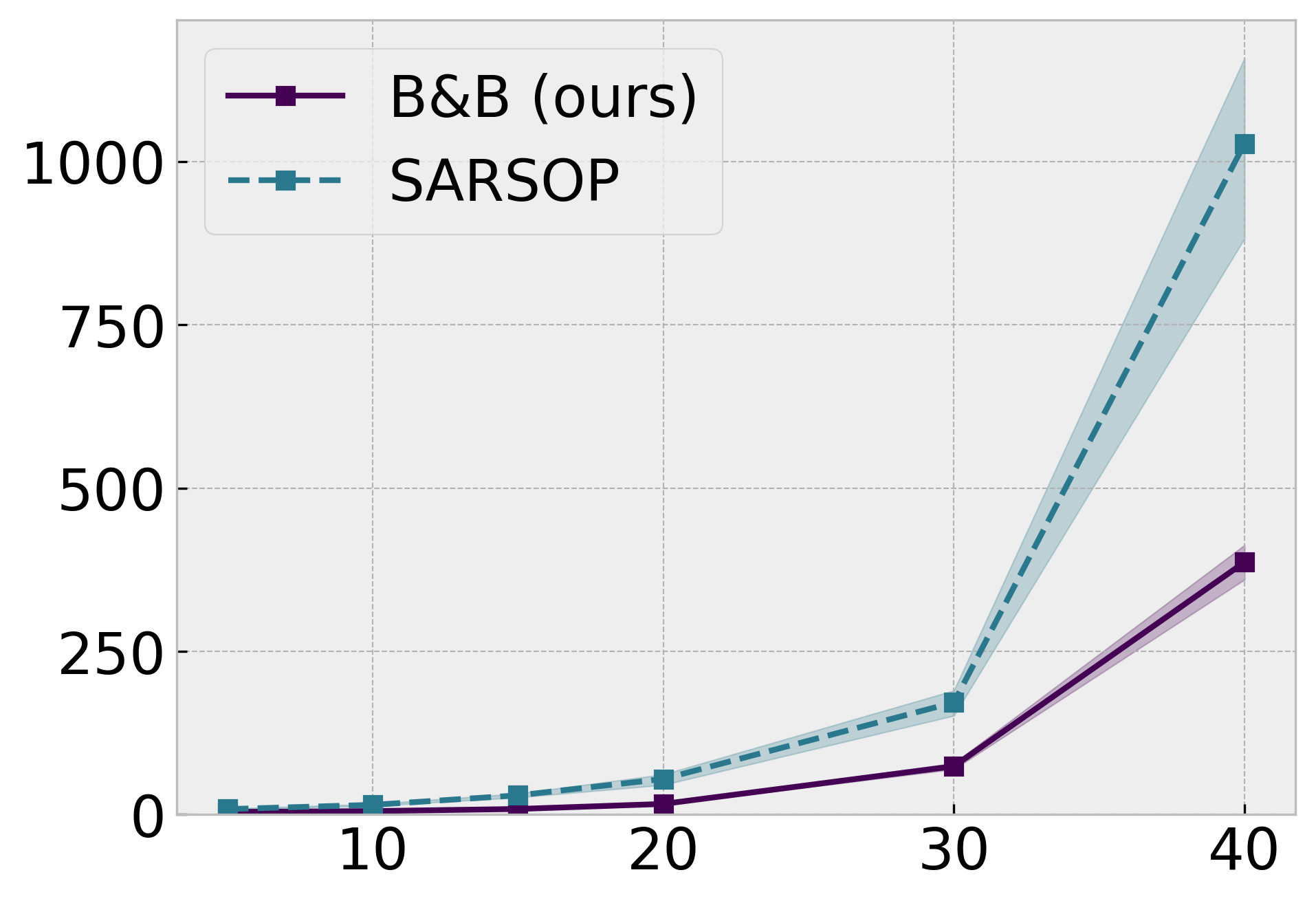}
   \caption{x-axis varies types and categories}
   \label{fig:actions_types_time}
   \vspace{0.5cm} 
\end{subfigure}%
\hspace{0.2cm} 
\begin{subfigure}{0.3\textwidth}
   \centering
   \includegraphics[width=\linewidth]{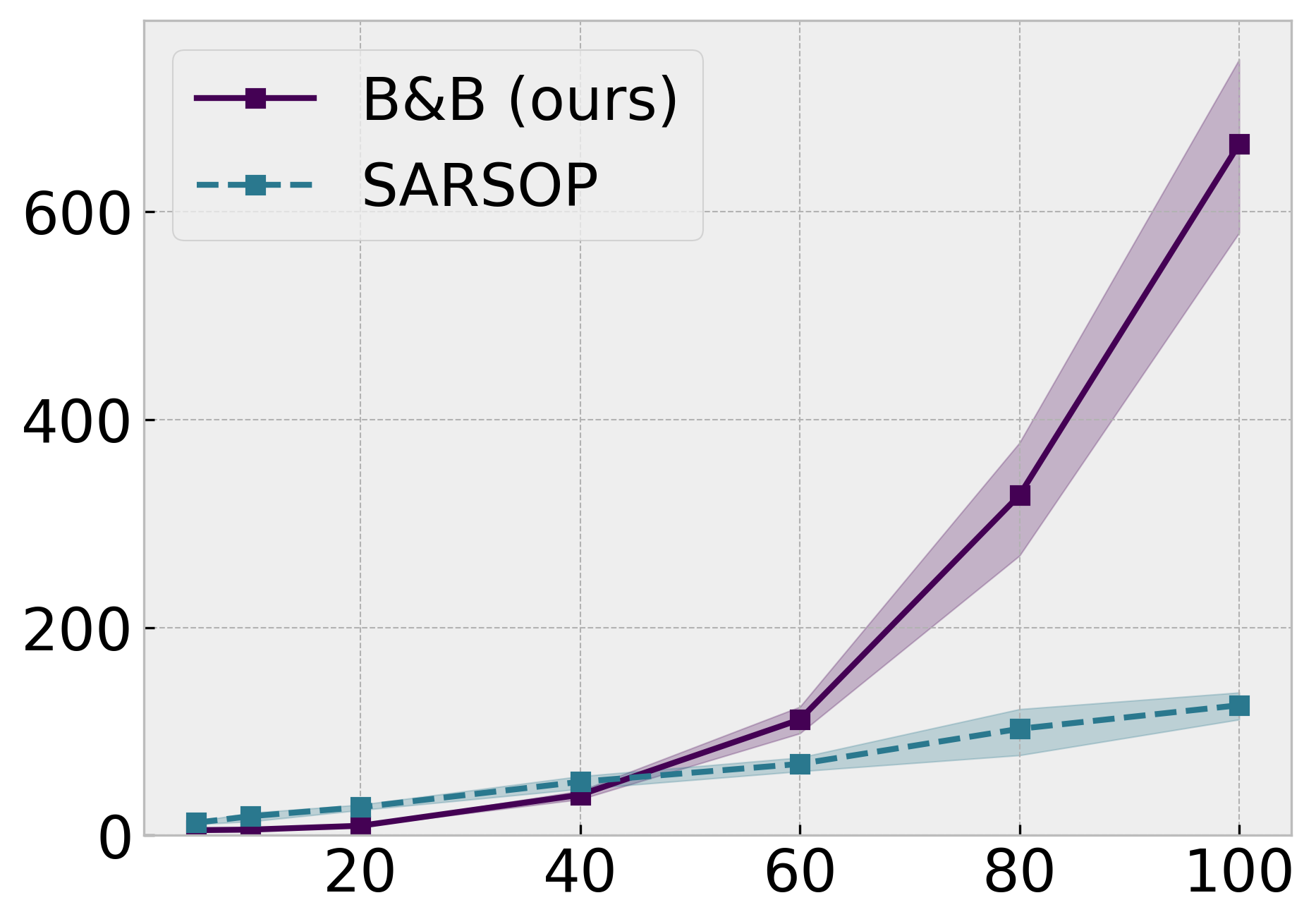}
   \caption{x-axis varies categories; 10 types}
   \label{fig:actions_time}
   \vspace{0.5cm} 
\end{subfigure}%
\caption{Runtime comparison (in milliseconds) between Algorithm~\ref{bb-algorithm} and a POMDP solver SARSOP. Each data point represents an average runtime over $500$ of $\prob$ instances. Shaded intervals represent $95\%$ bootstrap confidence intervals of the empirical average. Both algorithms stop when they reach a precision of $\varepsilon=10^{-6}$.\label{fig:sarsop}}
\end{figure*}

\paragraph{Runtime comparison}
Our model is novel, so there are no specifically tailored baselines. However, since it can be cast as a POMDP, we can compare Algorithm~\ref{bb-algorithm} with more general solvers. As a baseline, we have chosen SARSOP, a well-known offline point-based POMDP solver \cite{kurniawati2009sarsop}. While Algorithm~\ref{bb-algorithm} is straightforward, 
SARSOP is rather complex. It represents the optimal policy through $\alpha$-vectors (a convex piece-wise linear approximation of the value function) and clusters sampled beliefs to estimate the values of new ones. We used an open-source implementation of SARSOP,\footnote{\url{https://github.com/AdaCompNUS/sarsop}} and, for a fair comparison, implemented Algorithm~\ref{bb-algorithm} in the same language. 

Figure~\ref{fig:sarsop} presents the runtime comparison between Algorithm~\ref{bb-algorithm} and SARSOP. The statistical tests that support the comparison are deferred to \apxref{sec:auxiliary-details-about-the-experiments}. While Algorithm~\ref{bb-algorithm} dominates SARSOP on rectangular problems with a few categories (Figure~\ref{fig:types_time}) and overperforms SARSOP on square problems (Figure~\ref{fig:actions_types_time}), it underperforms when the number of categories is much higher than the number of user types (Figure~\ref{fig:actions_time}). 

We hypothesize that this result is due to SARSOP more effectively dealing with similar categories (similar associated rows in the matrix $\bm P$) through the $\alpha$-vector representation and clustering heuristic. Another explanation is that Algorithm~\ref{bb-algorithm} explores the policy space; hence, the branching factor is the number of categories. In contrast, SARSOP explores the belief space, whose dimension is the number of types. Consequently, we could expect SARSOP to struggle in cases with many user types and Algorithm~\ref{bb-algorithm} to encounter challenges in cases with many categories. Overall, each algorithm excels under different conditions.

\section{Experiments with Real-World Data} 
Beyond the simulations in Section~\ref{sec:experiments}, we also conducted experiments with the Movielens 1M dataset~\cite{harper2015movielens}. To demonstrate the applicability of our approach, we outline below how we transform sparse rating matrices into the required model parameters.

Real-world RS datasets typically consist of a sparse user-item rating matrix, where observed entries represent user ratings for items, with the majority of entries being unobserved. This representation differs from our model's requirements of a dense probability matrix between user types and content categories, accompanied by a prior distribution over user types; hence, our model cannot be applied directly, and the following two transformations are required.

First, the sparse rating matrix must be aggregated into a concise representation through clustering of users and items. This is a well-studied task and several solutions have been proposed in the literature, such as Spectral co-clustering~\cite{coclustering} or DBSCAN~\cite{dbscan}. Second, the clustered data must be transformed into model parameters. One straightforward way to achieve this is to construct the preference matrix by computing mean ratings within cluster pairs and normalizing to $[0,1]$, and deriving the prior distribution from cluster sizes. We defer this analysis to \apxref{sec:movielens}. Our empirical investigation using this real-world dataset substantiates the qualitative patterns observed in Section~\ref{sec:experiments}. Importantly, we stress that the qualitative results we obtain in this section extend beyond the synthetic setup to real-world datasets.


\section{Discussion and Future Work}\label{sec:disc}
We have introduced a model that captures two intertwined challenges: Aggregated user information and the risk of churn. We analyzed the belief walks and showed that the optimal policy must eventually act greedily. We have proposed a lower and upper bound on the optimal social welfare and demonstrated that our B\&B algorithm performs comparably to a state-of-the-art baseline.

We see considerable scope for future work. First, despite the nontrivial analysis, we still lack either a provably optimal polynomial-time algorithm or a formal proof of hardness. Second, future work could model more complex interactions where, e.g., users can dislike a category but continue the user session. In such a case, we can still apply Bayesian updates, but there is no clear notion of belief walks. This forms a technical challenge that is beyond the scope of our current paper (see \apxref{sec:model-extensions} for an elaborated discussion). 
Third, another challenge could stem from non-stationary preferences, namely if the matrix $\bm P$ changes over time or depends on the history of recommendations in the current session. 
Exploring these directions will further reveal how to effectively navigate the balance between exploration and exploitation in the face of aggregated user information.



\section*{Acknowledgements}
This research was supported by the Israel Science Foundation (ISF; Grant No. 3079/24). We thank Dmitry Ivanov for his helpful comments and assistance with the experimental part.

\bibliography{main}

\ifnum\Includeappendix=1{
\appendix
\onecolumn

\section{Showcasing Belief Walks under Various Policies} \label{sec:belief walks}
In this section, we illustrate the structure of belief walks under varying policies, showcasing the different behaviors and trade-offs that arise in the context of belief walks. Each graph corresponds to a different instance with three categories and three user types ($\abs{M} = \abs{K} = 3$). We provide the full details of the instances below. In each graph, we plot the belief walk induced by three policies. Each belief in the 3-dimensional simplex is characterized by a probability for the first type $\bm b(m_1)$, given in the x-axis, probability in the second type $\bm b(m_2)$, given in the y-axis, and probability for the third type $\bm b(m_3)$, given implicitly by $\bm b(m_3) = 1 -\bm b(m_1) - \bm b(m_2)$. The plotted policies are:
\begin{enumerate}
\item The optimal policy, calculated using Algorithm~\ref{bb-algorithm}.
\item The best fixed-action policy, which always recommends the category that maximizes the value function over all single-action policies w.r.t the prior. Put formally - \[ \pi^f = \left( \tilde k \right)_{t=1}^{\infty}, \quad \tilde k = \argmax_{k \in K} \sum_{m \in M} \bm{q}(m) \cdot \frac{\bm{P}(k, m)}{1 - \bm{P}(k, m)}. \]
For the rest of the section, we abbreviate "best fixed-action policy" to "BFA policy".
\item The myopic policy, which always recommends the category that yields the highest immediate reward in the current belief and then updates its belief afterward. Formally, \[ \pi^m = \left( k_t \right)_{t=1}^{\infty}: k_i = \argmax_{k \in K} \sum_{m \in M} \bm{b_i}(m) \cdot \bm{P}(k, m), \ \ \bm b_1 = \bm q, \ \ \forall i > 1: \bm b_i = \tau(\bm b_{i-1}, k_{i-1}). \]
\end{enumerate}

Each belief in the belief path produced by the optimal policy is assigned a numerical value based on its sequence position. The plots display these values as small black numbers beside some arrows. Each arrow signifies a Bayesian update followed by a recommendation, with the arrow's direction indicating the movement from the previous belief to the updated belief. The black dot depicts the initial prior $\bm q$ over the user type, marking the starting point for any plotted belief walk.

\begin{figure}
\centering
\begin{subfigure}{0.45\textwidth}
    \includegraphics[width=\linewidth]{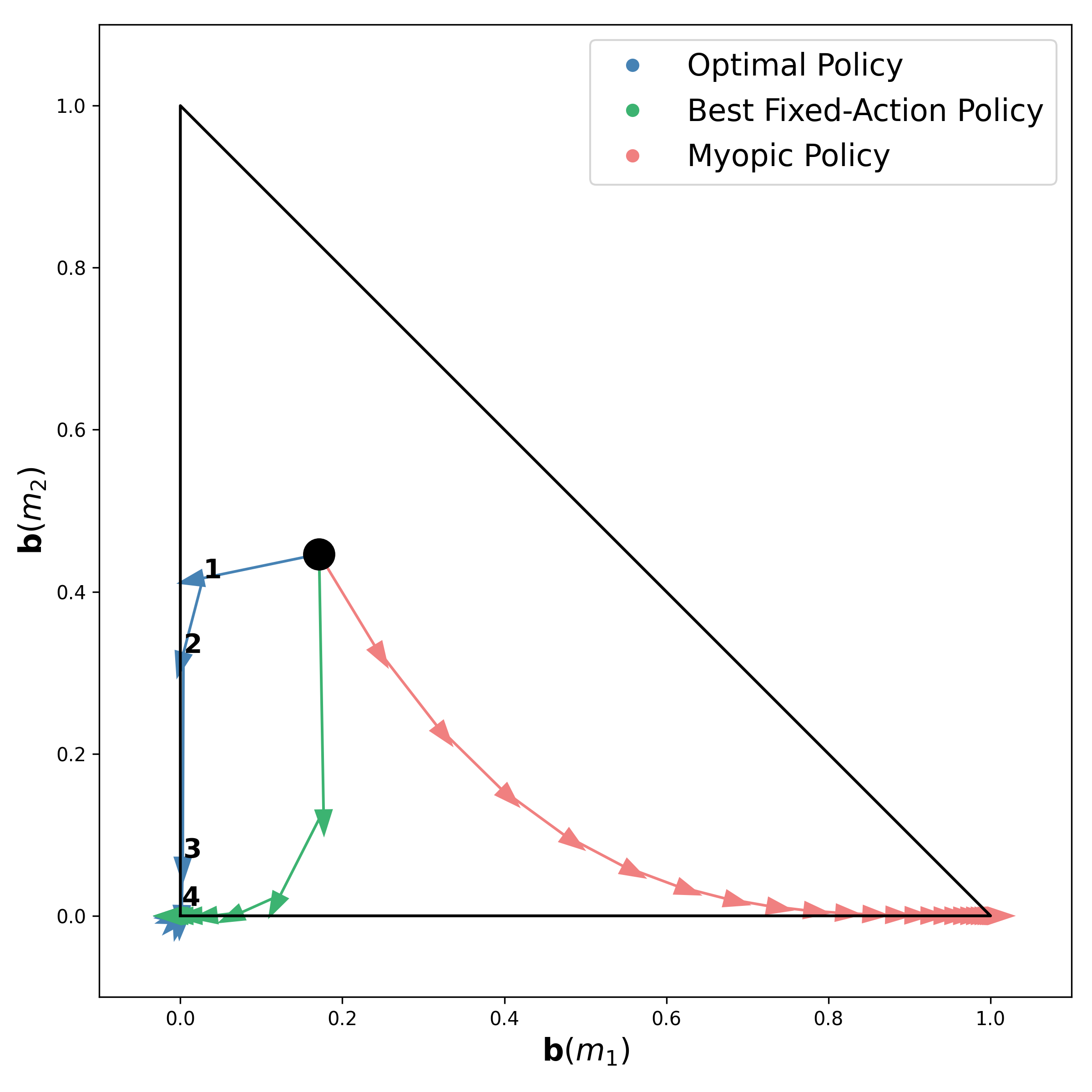}
    \caption{Figure accompanying Example~\ref{example:1} - the optimal policy avoids myopic behavior.}
    \label{fig:belief-walk1}
\end{subfigure}
\hfill
\begin{subfigure}{0.45\textwidth}
    \includegraphics[width=\linewidth]{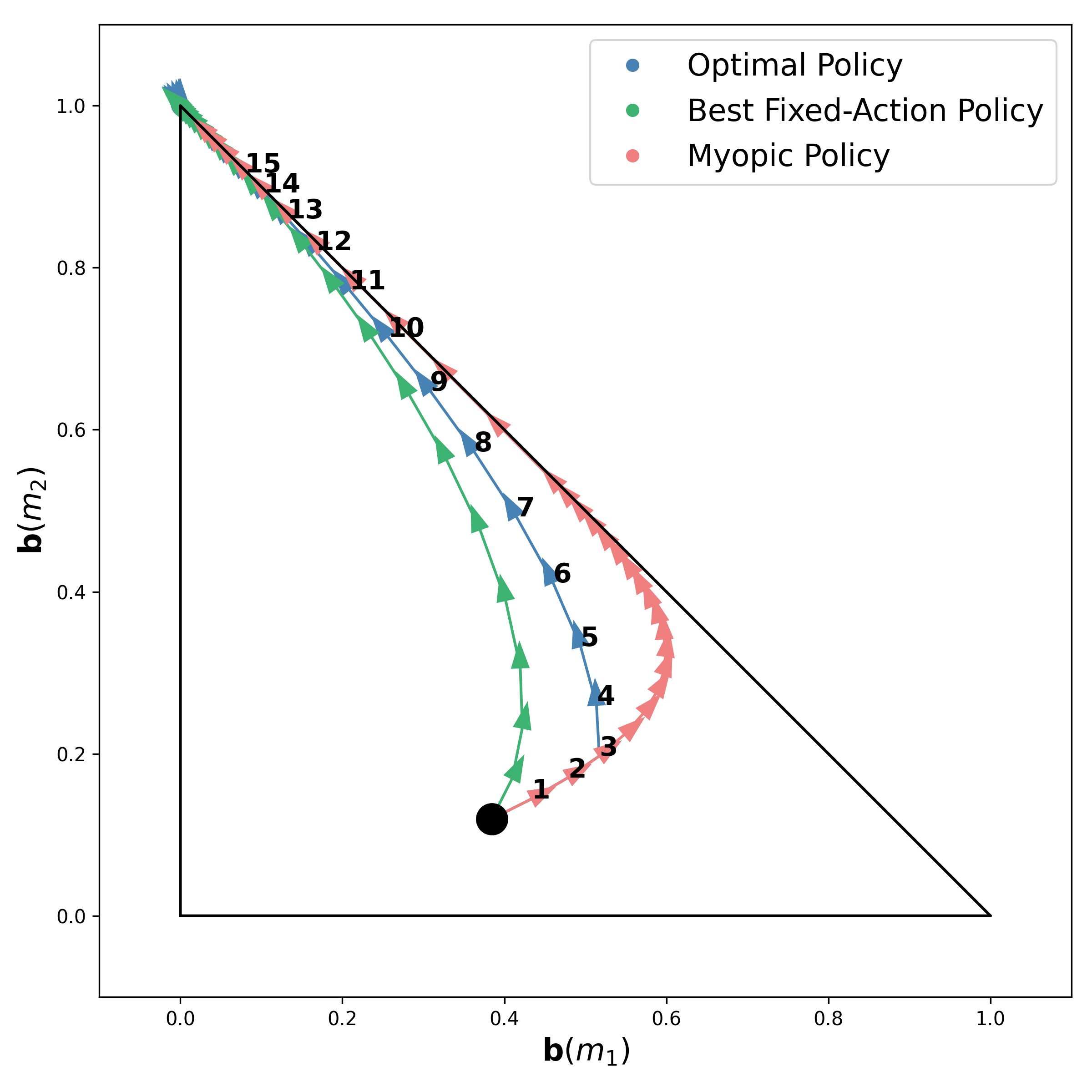}
    \caption{Figure accompanying Example~\ref{example:2} - all policies converge to the same vertex.}
    \label{fig:belief-walk2}
\end{subfigure}

\vspace{1em}

\begin{subfigure}{0.45\textwidth}
    \includegraphics[width=\linewidth]{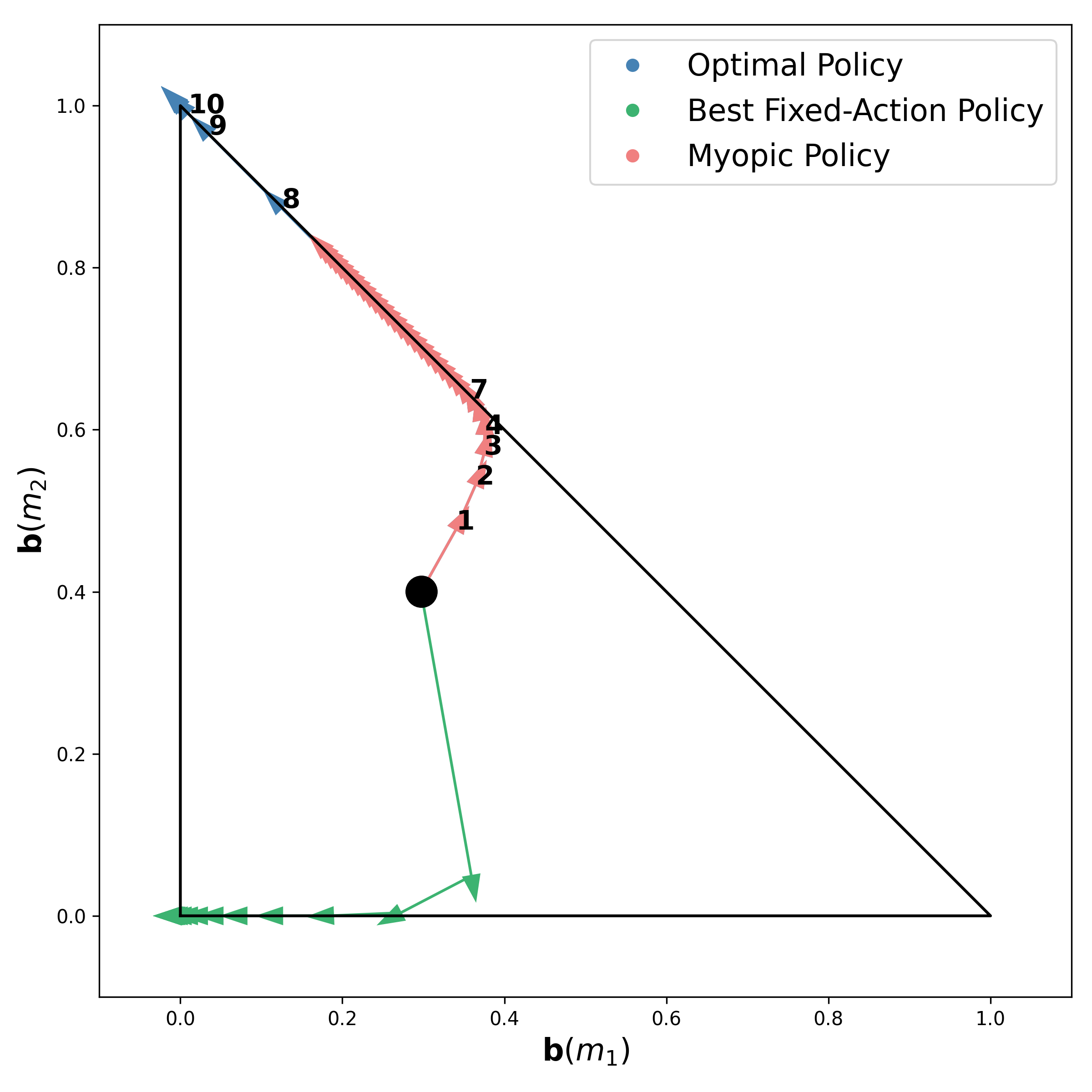}
    \caption{Figure accompanying Example~\ref{example:3} - the optimal policy diverges from the BFA policy prefix and forsakes the myopic policy for faster convergence.}
    \label{fig:belief-walk3}
\end{subfigure}
\hfill
\begin{subfigure}{0.45\textwidth}
    \includegraphics[width=\linewidth]{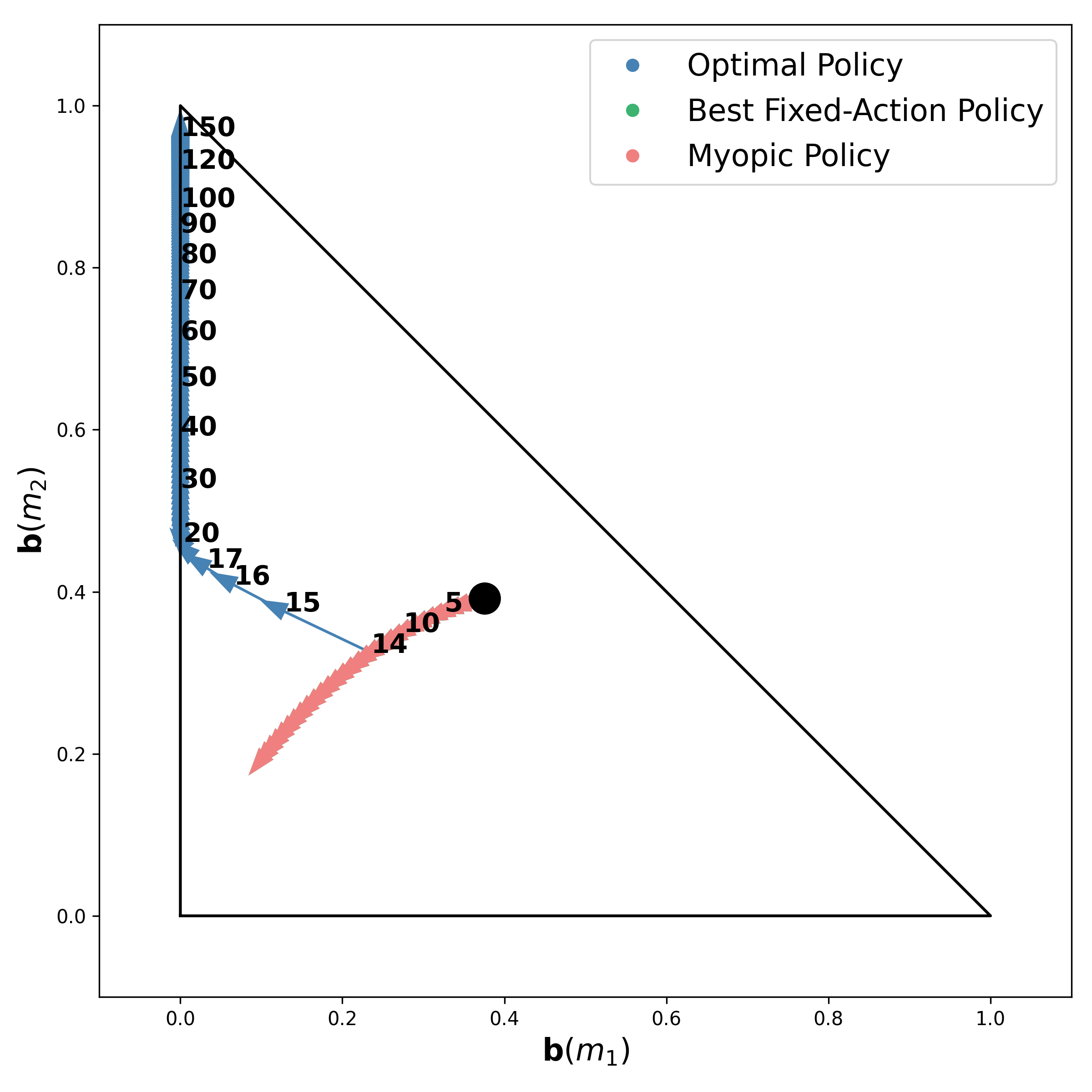}
    \caption{Figure accompanying Example~\ref{example:4} - there is a decrease followed by an increase in the belief of the type the belief walk converges to.}
    \label{fig:belief-walk4}
\end{subfigure}
\caption{Belief walks under different policies}
\label{fig:belief-walks}
\end{figure}

\begin{example} \label{example:1}
The $\prob$ instance in Figure~\ref{fig:belief-walk1} is $\langle K = \{ k_1, k_2, k_3 \},\ M = \{ m_1, m_2, m_3 \},\ \bm{q}_1,\ \bm{P}_1 \rangle$, where:
\begin{itemize}
\item $\bm{q}_1 = \begin{pmatrix} 0.1713 \\ 0.4465 \\ 0.3822 \end{pmatrix}$
\item $\bm{P}_1 = \begin{pmatrix} 0.8611 & 0.4591 & 0.6862 \\ 0.0969 & 0.5531 & 0.8604 \\ 0.5055 & 0.1430 & 0.8879 \\ \end{pmatrix}$
\end{itemize}
In this instance, the optimal and BFA policies converge to the user type $m_3$, whereas the myopic policy converges to $m_1$. This observation underlines the central idea of Proposition~\ref{prop:myopic-policy-suboptimality}, which asserts that the myopic policy may overlook future rewards, sometimes leading to convergence to a suboptimal user type. Even though the optimal and BFA policies reach the same user type, their pathways differ. This underscores the importance of meticulous initial exploration, as the sequence of recommendations before convergence significantly influences the overall rewards.
\end{example}

\begin{example} \label{example:2}
The $\prob$ instance in Figure~\ref{fig:belief-walk2} is $\langle K = \{ k_1, k_2, k_3 \},\ M = \{ m_1, m_2, m_3 \},\ \bm{q}_2,\ \bm{P}_2 \rangle$, where:
\begin{itemize}
\item $\bm{q}_2 = \begin{pmatrix} 0.3844 \\ 0.1197 \\ 0.4959 \end{pmatrix}$
\item $\bm{P}_2 = \begin{pmatrix} 0.6848 & 0.9100 & 0.5457 \\ 0.7741 & 0.8284 & 0.5833 \\ 0.1931 & 0.9127 & 0.5273 \\ \end{pmatrix}$
\end{itemize}
In this instance, all three policies converge to $m_2$, albeit through different paths. Initially, for the first three recommendations, the optimal policy aligns with the myopic policy. Later, it diverges to secure higher future rewards, considering that it might get lower immediate rewards.
\end{example}

\begin{example} \label{example:3}
The $\prob$ instance in Figure~\ref{fig:belief-walk3} is $\langle K = \{ k_1, k_2, k_3 \},\ M = \{ m_1, m_2, m_3 \},\ \bm{q}_3,\ \bm{P}_3 \rangle$, where:
\begin{itemize}
\item $\bm{q}_3 = \begin{pmatrix} 0.2972 \\ 0.4001 \\ 0.3027 \end{pmatrix}$
\item $\bm{P}_3 = \begin{pmatrix} 0.5492 & 0.0560 & 0.8878 \\ 0.2195 & 0.8576 & 0.2072 \\ 0.7674 & 0.7992 & 0.4051 \\ \end{pmatrix}$
\end{itemize}
In this instance, the BFA policy deviates from optimal and myopic policies early. Initially, these two policies align in their recommendations until the optimal policy transitions from belief (7) to belief (8), as shown in the graph. Although this recommendation provides a lower immediate reward than the myopic recommendation, it leads to a belief state much closer to $m_2$ than the updated belief after the myopic recommendation. This emphasizes the long-term rewards in this scenario, highlighting the common trade-off between short-term and long-term goals in sequential decision-making.
\end{example}

\begin{example} \label{example:4}
The $\prob$ instance in Figure~\ref{fig:belief-walk4} is $\langle K = \{ k_1, k_2, k_3 \},\ M = \{ m_1, m_2, m_3 \},\ \bm{q}_4,\ \bm{P}_4 \rangle$, where:
\begin{itemize}
\item $\bm{q}_4 = \begin{pmatrix} 0.3755 \\ 0.3921 \\ 0.2324 \end{pmatrix}$
\item $\bm{P}_4 = \begin{pmatrix} 0.4011 & 0.8521 & 0.8301 \\ 0.7683 & 0.7837 & 0.8314 \\ 0.7674 & 0.7832 & 0.4051 \\ \end{pmatrix}$
\end{itemize}
In this instance, the BFA and myopic policies coincide and act optimally until the 14th step, where they both differ from the optimal policy. This figure is an example of the phenomenon described in Section~\ref{sec:experiments}—in the first 14 steps, the entry in the belief walk that corresponds to user type $m_2$ decreases, and one could expect the belief walk to approach a different vertex. Then, in step 15, as the optimal policy diverges from the other two policies, we see a sudden increase in $\bm{b}(m_2)$, which persists throughout the remainder of the belief walk. In conclusion, it is possible to see a decrease before the increase in the certainty of the user type to which the belief walk finally converges.
\end{example}

\section{Omitted Proofs from Section \ref{sec:problem-definition}}
\label{sec:problem-definition-proofs}

\begin{proof}[\normalfont\bfseries Proof of Lemma~\ref{lemma:closed-form-representations-of-the-value-function}]
    The first expression can be derived through Observation~\ref{obs:recursive-formula-of-the-value-function}.
    First we prove via induction that for every $n \in \mathbb{N}$ the value function can be represented as $V^{\pi}(\bm{b}) = \sum_{t=1}^{n} \prod_{j=1}^{t} p_{\pi_{j}}(\bm{b}^{\pi, \bm{b}}_j) + \prod_{t=1}^{n} p_{\pi_{t}}(\bm{b}^{\pi, \bm{b}}_t) \cdot V^{\pi[n+1:]}(\bm{b}^{\pi, \bm{b}}_{n+1})$. Then, as each $p_{\pi_{t}}(\bm{b}^{\pi, \bm{b}}_t)$ is a probability smaller than $1$ and the value function is bounded from above by $\frac{p_{\max}}{1-p_{\max}}$, as $n \to \infty$ the expression $\prod_{t=1}^{n} p_{\pi_{t}}(\bm{b}^{\pi, \bm{b}}_t) \cdot V^{\pi[n+1:]}(\bm{b}^{\pi, \bm{b}}_{n+1})$ will decay to $0$, and the value function will converge to the first expression.

    \textbf{Base case:} For $n=1$ it holds directly from Observation~\ref{obs:recursive-formula-of-the-value-function}.

    \textbf{Inductive step:} Assume that the statement holds for $n$, and prove it for $n+1$.
    \begin{align*}
        V^{\pi}(\bm{b}) & = p_{\pi_{1}}(\bm{b})(1 + V^{\pi[2:]}_n(\bm{b}^{\pi, \bm{b}}_{2}))    \\
        & = p_{\pi_{1}}(\bm{b}) + p_{\pi_{1}}(\bm{b}) \cdot \sum_{t=1}^{n} \prod_{j=1}^{t} p_{\pi_{j+1}}(\bm{b}^{\pi, \bm{b}}_{j+1}) + p_{\pi_{1}}(\bm{b}) \cdot \prod_{t=1}^{n} p_{\pi_{t+1}}(\bm{b}^{\pi, \bm{b}}_{t+1}) \cdot V^{\pi[n+1:]}(\bm{b}^{\pi, \bm{b}}_{n+1}) \\
        & = \sum_{t=1}^{n+1} \prod_{j=1}^{t} p_{\pi_{j}}(\bm{b}^{\pi, \bm{b}}_j) + \prod_{t=1}^{n+1} p_{\pi_{t}}(\bm{b}^{\pi, \bm{b}}_t) V^{\pi[n+1:]}(\bm{b}^{\pi, \bm{b}}_{n+1}).  
    \end{align*}
    Where the first equality is due to Observation~\ref{obs:recursive-formula-of-the-value-function} and the second equality follows from applying the induction hypothesis to $V^{\pi[2:]}_n(\bm{b}^{\pi, \bm{b}}_{2})$.

    For the second expression, we can calculate the expected social welfare via conditioning on the sampled user type.
    \[
    V^\pi(\bm b) = \sum_{m \in M}\bm b(m) \E{F(m,\pi)}.
    \]
    Notice that the expected utility from following policy $\pi$ for a user of type $m \in M$ is:
    \begin{align*}
        \E{F(m,\pi)} &= \sum_{t=1}^{\infty} \mathbb{P}(\text{User of type } m \text{ hadn't departed by time } t \text{ when following a policy } \pi) \\
        &\cdot \mathbb{P}(\text{User of type } m \text{ liked the recommendation in time } t) \\
        &= P(\pi_1, m) + \sum_{t=2}^{\infty} \prod_{j=1}^{t-1} \bm P(\pi_j, m) \cdot P(\pi_t, m) \\
        &= \sum_{t=1}^{\infty} \prod_{j=1}^{t} \bm{P}(\pi_j, m).
    \end{align*}
    Together, the expected social welfare can be calculated as
    \[
        V^\pi(\bm b) = \sum_{m} \bm{b}(m) \sum_{t=1}^{\infty} \prod_{j=1}^{t} \bm{P}(\pi_j, m).
    \]
This completes the proof of the lemma.
\end{proof}

\begin{proof}[\normalfont\bfseries Proof of Proposition~\ref{prop:myopic-policy-suboptimality}]
    Set some arbitrary $d \in \mathbb{R_+}$, and consider the following $\prob$ instance:
    \[
        \left\langle
        K = \{ k_1, k_2 \}, \quad M = \{ m_1, m_2 \},
        \bm q = \begin{pmatrix}
            0.5 \\
            0.5
        \end{pmatrix},
        \bm P = \begin{pmatrix}
            \frac{8d}{1+8d} & 0   \\
            0.8             & 0.8
        \end{pmatrix}
        \right\rangle.
    \]
    In this instance, the myopic policy will always choose $k_2$, as it has a higher immediate reward, and as $\bm{P}(k_2, m_1) = \bm{P}(k_2, m_2)$, the belief will stay the same after the Bayesian update.
    Calculating the expected number of likes for the myopic policy using the second expression in Lemma~\ref{lemma:closed-form-representations-of-the-value-function} yields:
    \[
        V^{\pi^m} = 0.5 \cdot \frac{0.8}{1-0.8} + 0.5 \cdot \frac{0.8}{1-0.8} = 4.
    \]
    Now, consider the policy $\pi^{k_1}$ that chooses $k_1$ indefinitely.
    Calculating the expected number of likes for this policy, we get
    \[
        V^{\pi^{k_1}} = 0.5 \cdot \frac{\frac{8d}{1+8d}}{1-\frac{8d}{1+8d}} + 0.5 \cdot 0 = 0.5 \cdot \frac{8d}{1+8d-8d} + 0.5 \cdot 0 = 4d.
    \]
    Because $\pi^\star$ is optimal we get that $V^\star \geq V^{\pi^{k_1}} = 4d = d \cdot V^{\pi^m}$, and that concludes the proof.
\end{proof}

\section{Omitted Proofs from Section \ref{sec:approximating-the-optimal-policy}}
\label{sec:approximating-the-optimal-policy-proofs}

\begin{proof}[\normalfont\bfseries Proof of Lemma~\ref{lemma:finite-horizon-approximation}]
  As can be seen in the proof Lemma~\ref{lemma:closed-form-representations-of-the-value-function}, The value function can be represented as
  \[
    V^{\pi}(\bm{b}) = \sum_{t=1}^{H(\varepsilon)} \prod_{j=1}^{t} p_{\pi_{j}}(\bm{b}^{\pi, \bm{b}}_j) + \prod_{t=1}^{H(\varepsilon)} p_{\pi_{t}}(\bm{b}^{\pi, \bm{b}}_t) \cdot V^{\pi[H(\varepsilon)+1:]}(\bm{b}^{\pi, \bm{b}}_{H(\varepsilon)+1}).
  \]
  Notice that for every $L \in \mathbb{N}$, $\prod_{t=1}^{L} p_{\pi_{t}}(\bm{b}^{\pi, \bm{b}}_t) \leq p_{\max}^L$. Additionally, as stated in Equation~\ref{eq:bounded-value}, for every belief $b$ and policy $\pi$ we have that $V^{\pi}(\bm{b}) \leq \frac{p_{\max}}{1-p_{\max}}$.
  Therefore:
  \[
    V^{\pi^\star}(\bm{b}) \leq \sum_{t=1}^{H(\varepsilon)} \prod_{j=1}^{t} p_{\pi^\star_{j}}(\bm{b}^{\pi^\star, b}_j) + p_{\max}^{H(\varepsilon)} \cdot \frac{p_{\max}}{1-p_{\max}} \leq \sum_{t=1}^{H(\varepsilon)} \prod_{j=1}^{t} p_{\pi^\star_{j}}(\bm{b}^{\pi, \bm{b}}_j) + \varepsilon = V_{H(\varepsilon)}^{\pi^\star}(\bm{b}) + \varepsilon \leq \max_{\pi} V_{H(\varepsilon)}^{\pi}(\bm{b}) + \varepsilon.
  \]
\end{proof}

\begin{lemma}
  \label{lemma:transitivity-of-belief-updates}
  Let $\langle M, K, \bm{q}, \bm{P} \rangle$ be an instance of the problem, a belief $\bm{b} \in \Delta(M)$, and 2 categories $k_1, k_2 \in K$. Then:
  \[
    \tau(\tau(\bm{b}, k_1), k_2) = \tau(\tau(\bm{b}, k_2), k_1)
  \]
\end{lemma}

\begin{proof}[\normalfont\bfseries Proof of Lemma~\ref{lemma:transitivity-of-belief-updates}]
 Using the definition of the Bayesian update, the following holds for every $m \in M$:
 \begin{align*}
   \tau(\tau(\mathbf{b}, k_1), k_2)(m) & = \frac{\mathbf{P}(k_2, m) \cdot \tau(\mathbf{b}, k_1)(m)}{\sum_{m' \in M} \mathbf{P}(k_2, m') \cdot \tau(\mathbf{b}, k_1)(m')} \\
   & = \frac{\mathbf{P}(k_2, m) \cdot \frac{\mathbf{P}(k_1, m) \cdot \mathbf{b}(m)}{\sum_{m' \in M} \mathbf{P}(k_1, m') \cdot \mathbf{b}(m')}}{\sum_{m' \in M} \mathbf{P}(k_2, m') \cdot \frac{\mathbf{P}(k_1, m') \cdot \mathbf{b}(m')}{\sum_{m'' \in M} \mathbf{P}(k_1, m'') \cdot \mathbf{b}(m'')}} \\
   & = \frac{\mathbf{P}(k_2, m) \cdot \mathbf{P}(k_1, m) \cdot \mathbf{b}(m)}{\sum_{m' \in M} \mathbf{P}(k_2, m') \cdot \mathbf{P}(k_1, m') \cdot \mathbf{b}(m')} \\
   & = \frac{\mathbf{P}(k_1, m) \cdot \mathbf{P}(k_2, m) \cdot \mathbf{b}(m)}{\sum_{m' \in M} \mathbf{P}(k_1, m') \cdot \mathbf{P}(k_2, m') \cdot \mathbf{b}(m')} \\
   & = \frac{\mathbf{P}(k_1, m) \cdot \frac{\mathbf{P}(k_2, m) \cdot \mathbf{b}(m)}{\sum_{m' \in M} \mathbf{P}(k_2, m') \cdot \mathbf{b}(m')}}{\sum_{m' \in M} \mathbf{P}(k_1, m') \cdot \frac{\mathbf{P}(k_2, m') \cdot \mathbf{b}(m')}{\sum_{m'' \in M} \mathbf{P}(k_2, m'') \cdot \mathbf{b}(m'')}} \\
   & = \frac{\mathbf{P}(k_1, m) \cdot \tau(\mathbf{b}, k_2)(m)}{\sum_{m' \in M} \mathbf{P}(k_1, m') \cdot \tau(\mathbf{b}, k_2)(m')} \\
   & = \tau(\tau(\mathbf{b}, k_2), k_1)(m).
 \end{align*}
This completes the proof of the lemma.
\end{proof}

\begin{algorithm}
   \caption{Dynamic Programming Approximation \label{alg:backward-induction-approximation}}
   \begin{algorithmic}[1]
     \REQUIRE An instance of the problem $\langle M, K, \bm{q}, \bm{P} \rangle$, a horizon $H \in \mathbb{N}$, and an error term $\varepsilon$.
     \ENSURE $BT(q, 0)$ - an approximation of the optimal policy $\pi$.
     \STATE Initialize dicts $V: \Delta(M) \times \{0, 1, \ldots, H\} \to \mathbb{R}$, $BT: \Delta(M) \times \{0, 1, \ldots, H\} \to \bigcup_{i=1}^{H} K^i$.
     \STATE Initialize $V(b, H+1) = 0$ for every $\bm{b} \in \Delta(M) \text{ s.t } \exists \{ n_i \}_{i=1}^{\abs{K}}: \sum_{i=1}^{\abs{K}} n_i = H+1, \tau(q, k_1^{n_1}, k_2^{n_2}, \ldots, k_{\abs{K}}^{n_{\abs{K}}}) = b$.
     \STATE Initialize $BT(b, H+1) = \emptyset$ for every $\bm{b} \in \Delta(M) \text{ s.t } \exists \{ n_i \}_{i=1}^{\abs{K}}: \sum_{i=1}^{\abs{K}} n_i = H+1, \tau(q, k_1^{n_1}, k_2^{n_2}, \ldots, k_{\abs{K}}^{n_{\abs{K}}}) = b$.
     \FOR{$h = H, ..., 1$}
     \FOR{every $\bm{b} \in \Delta(M) \text{ s.t } \exists \{ n_i \}_{i=1}^{\abs{K}}: \sum_{i=1}^{\abs{K}} n_i = h, \tau(q, k_1^{n_1}, k_2^{n_2}, \ldots, k_{\abs{K}}^{n_{\abs{K}}}) = b$}
     \STATE $V(b, h) = \max_{j \in [K]: n_j > 0} \left\{ p_k \cdot \left( 1 + V(\tau(\bm{b}, k_1^{n_1}, k_2^{n_2}, \ldots, k_j^{n_j-1}, \ldots, k_{\abs{K}}^{n_{\abs{K}}}), h+1) \right) \right\}$.
     \STATE $BT(b, h) = \argmax_{j \in [K]: n_j > 0} \left\{ p_k \cdot \left( 1 + V(\tau(\bm{b}, k_1^{n_1}, k_2^{n_2}, \ldots, k_j^{n_j-1}, \ldots, k_{\abs{K}}^{n_{\abs{K}}}), h+1) \right) \right\} \oplus BT(\tau(\bm{b}, k_1^{n_1}, k_2^{n_2}, \ldots, k_j^{n_j-1}, \ldots, k_{\abs{K}}^{n_{\abs{K}}}), h+1)$.
     \ENDFOR
     \ENDFOR
     \STATE $V(q, 0) = \max_{j \in [K]} \left\{ p_k \cdot \left( 1 + V(\tau(q, k_j), 1) \right) \right\}$.
     \STATE $BT(q, 0) = \argmax_{j \in [K]} \left\{ p_k \cdot \left( 1 + V(\tau(q, k_j), 1) \right) \right\} \oplus BT(\tau(q, k_j), 1)$.
   \end{algorithmic}
\end{algorithm}
  
\begin{proof}[\normalfont\bfseries Proof of Proposition~\ref{prop:backward-induction-approximation}]
  We begin by mentioning two conclusions of Lemma~\ref{lemma:transitivity-of-belief-updates}:
  \begin{enumerate}
    \item The order in which we recommend items only affects the accumulated reward, not the belief the system transitioned into.
    \item The number of beliefs that can be reached after $L$ recommendations is at most ${ L + \abs{K} - 1 \choose \abs{K} - 1}$, as the only thing that effects the belief is the number of times each item was recommended (the multiset of categories rather than the sequence).
  \end{enumerate}

  Following those conclusions, we can define a dynamic programming algorithm for solving the finite horizon recommendation problem and analyze its complexity.
  
  For ease of notation, we will denote the composition of $L$ belief updates with the categories \\
  $k_1^1, \ldots, k_1^{n_1}, k_2^1, \ldots, k_2^{n_2}, \ldots, k_{\abs{K}}^1, \ldots, k_{\abs{K}}^{n_{\abs{K}}}$ (where $n_1 + n_2 + \ldots + n_{\abs{K}} = L$) as $\tau(\bm{b}, k_1^{n_1}, k_2^{n_2}, \ldots, k_{\abs{K}}^{n_{\abs{K}}})$. Remember that the order of the categories does not affect the updated belief, so this is well-defined.

  Algorithm~\ref{alg:backward-induction-approximation} takes a bottom-up approach - it starts by calculating the (low-depth) optimal value and policy for every belief that can be encountered at the most advanced layers and then uses Observation \ref{obs:recursive-formula-of-the-value-function} to calculate the optimal values and policies of every belief in the above layer. The process concludes upon arriving at layer $0$ (the root), where the optimal policy, which the algorithm aims to determine, is found. This methodology represents a significant enhancement over the brute-force approach, as it efficiently eliminates sub-optimal prefixes (those dominated by the identical set of actions arranged differently) in early stages, thereby avoiding unnecessary search operations.

  To establish that this algorithm achieves the optimal $H$-horizon value, we proceed by backward induction. Let $V^*_h(\bm{b})$ denote the optimal value achievable from belief $\bm{b}$ with exactly $h$ steps remaining. For the base case $h=0$, we have $V^*_0(\bm{b}) = 0$ for all beliefs $\bm{b}$, which is correctly initialized in the algorithm. For $h > 0$, the Bellman optimality equation gives us:
  
  \[
    V^*_h(\bm{b}) = \max_{k \in K} \left\{ p_k(\bm{b}) \cdot (1 + V^*_{h-1}(\tau(\bm{b}, k))) \right\}.
  \]

  Our algorithm computes precisely these values through its dynamic programming updates. Furthermore, for any belief $\bm{b}$ that can be reached after $t$ recommendations, the algorithm maintains the invariant that $V(\bm{b}, H-t)$ equals $V^*_{H-t}(\bm{b})$. This holds because:
  \begin{enumerate}
    \item The algorithm considers all possible categories at each step.
    \item It uses the correct Bayesian update rule for belief transitions.
    \item It applies the precise recursive formulation from Observation \ref{obs:recursive-formula-of-the-value-function}.
  \end{enumerate}

  Consequently, when the algorithm terminates, $V(\bm{q}, H)$ equals $V^*_H(\bm{q})$, yielding the optimal $H$-horizon value and its corresponding policy.

  Notice that the complexity of depth $h$ in this algorithm is $O\left( { h + \abs{K} - 1 \choose \abs{K} - 1} \cdot \abs{K} \cdot \abs{M} \right)$, as for every belief $b$ that can be reached after $h$ recommendations, we take a maximum over at most $\abs{K}$ categories, and for each category, we calculate the value of the next belief, which is a linear operation in $\abs{M}$.
  Therefore, the total complexity of this algorithm is $O\left( (H + \abs{K})^{\abs{K}} \cdot \abs{K} \cdot \abs{M} \right)$.
\end{proof}



\begin{proof}[\normalfont\bfseries Proof of Proposition~\ref{prop:discretization-approximation}]
  First, we note that the finite horizon version of our problem can be represented as a POMDP with a finite horizon, where the horizon is the maximal number of recommendations until the session terminates.
  Furthermore, the following build transforms our model, which doesn't have an explicit discount factor, into a POMDP with an explicit discount factor. The main trick here is to use $p_{\max} < 1$ as the discount factor and increase the probabilities in $\bm P$ by a factor of $\frac{1}{p_{\max}}$.
  
  There is a need for distinction between the first round, in which the reward isn't discounted, and the rest of the rounds, in which the rewards are discounted. We do that by introducing $2$ states, $m_s, m_f$, for every user type $m$. This distinction is only to handle the discounting - we still don't allow a change in user type. This will translate to the POMDP transitions model where we allow transitions only from $m_s$ to $m_f$ and to an absorbing state, and then transitions only from $m_f$ to $m_f$ and the absorbing state. For example, one cannot transition between $m_f$ to $m'_f$ for user types $m \neq m'$. 
  
  The observations in this POMDP model will be redundant - as long as the system hasn't transitioned into the absorbing state (the user is still in the system), the received observation will be "like." Otherwise, the received observation will be "dislike".

  Now we will formally describe this transformation. Given a $\prob$ instance $\langle M, K, \bm{q}, \bm{P} \rangle$, construct the following POMDP:
  \begin{itemize}
    \item The states are $\{ m_{s} \}_{m \in M} \cup \{ m_{f} \}_{m \in M} \cup \{ \bot \}$. The distinction between $m_{s}$ and $m_{f}$ is used to differentiate between the first round, in which the reward isn't discounted, and the rest of the rounds, in which the rewards are discounted.
    \item The actions are $K$.
    \item The transition function is defined as follows:
          \begin{itemize}
            \item $\bm{P}(k, m_{s}, m_{f}) = \bm{P}(k, m)$ for every $m \in M, k \in K$.
            \item $\bm{P}(k, m_{s}, \bot) = 1 - \bm{P}(k, m)$ for every $m \in M, k \in K$.
            \item $\bm{P}(k, m_{f}, m_{f}) = \frac{\bm{P}(k, m)}{p_{\max}}$ for every $m \in M, k \in K$.
            \item $\bm{P}(k, m_{f}, \bot) = 1 - \frac{\bm{P}(k, m)}{p_{\max}}$ for every $m \in M, k \in K$.
            \item $\bm{P}(k, \bot, \bot) = 1$ for every $k \in K$.
          \end{itemize}
    \item The reward function is defined as follows:
          \begin{itemize}
            \item $r(k, m_{s}, m_{f}) = 1$ for every $m \in M, k \in K$.
            \item $r(k, m_{s}, \bot) = 0$ for every $m \in M, k \in K$.
            \item $r(k, m_{f}, m_{f}) = 1$ for every $m \in M, k \in K$.
            \item $r(k, m_{f}, \bot) = 0$ for every $m \in M, k \in K$.
            \item $r(k, \bot, \bot) = 0$ for every $k \in K$.
          \end{itemize}
    \item The discount factor is $\gamma = p_{\max}$.
    \item The observations are $"like"$ and $"dislike"$.
    \item The observation model is defined as:
          \begin{itemize}
            \item $\bm{P}("like" | k, m_{s}, m_{f}) = 1$ for every $m \in M, k \in K$.
            \item $\bm{P}("dislike" | k, m_{s}, \bot) = 1$ for every $m \in M, k \in K$.
            \item $\bm{P}("like" | k, m_{f}, m_{f}) = 1$ for every $m \in M, k \in K$.
            \item $\bm{P}("dislike" | k, m_{f}, \bot) = 1$ for every $m \in M, k \in K$.
            \item $\bm{P}("dislike" | k, \bot, \bot) = 1$ for every $k \in K$.
          \end{itemize}
    \item The initial belief is $\forall m \in M: b(m_s) = \bm{q}(s), b(m_f) = 0, b(\bot)=0$.
  \end{itemize}

  As in the $\prob$ model, the "dislike" observation ends the episode and the reward is accumulated only for the "like" observations.

  While the constructed POMDP operates over a belief simplex of dimension $2|M|+1$ (accounting for both initial and non-initial states, plus the absorbing state), the reachable belief space is substantially more constrained. By construction, beliefs cannot simultaneously assign non-zero probability to both initial and non-initial variants of any user type, and the absorbing state, once reached, becomes a singular point in the belief space. Consequently, the effective dimensionality of the reachable belief space remains comparable to that of an $|M|$-dimensional simplex, preserving the theoretical guarantees without incurring additional computational complexity from the expanded state space.

  To show the equivalency between the optimal policy in our model and the one in the above POMDP, we will show that both the expected cumulative reward and the updated belief after performing the same sequence of actions are the same. This will be sufficient as those are the components that uniquely define optimal policy behavior.

  Let $k_1, \ldots, k_n$ be a set of actions. In our model, the expected accumulated reward after performing these actions is
  \[
    V = \sum_{m \in M} b(m) \cdot \sum_{t=1}^{n} \prod_{j=1}^{t} P(k_t, m).
  \]
  In the discussed POMDP model, the value will be as follows:
  \begin{align*}
    V & = \sum_{m \in M} \bm b(m) \cdot ( \bm P(k_1, m) + \bm P(k_1, m) \cdot p_{\max} \cdot \frac{\bm P(k_2, m)}{p_{\max}} + \bm P(k_1, m) \cdot \frac{\bm P(k_2, m)}{p_{\max}} \cdot p_{\max}^2 \cdot \frac{\bm P(k_3, m)}{p_{\max}} + \ldots \\
      & + \bm P(k_1, m) \cdot \frac{\bm P(k_2, m)}{p_{\max}} \cdot \frac{\bm P(k_3, m)}{p_{\max}} \cdot \ldots \cdot \frac{\bm P(k_{n-1}, m)}{p_{\max}} \cdot p_{\max}^{n-1} \cdot \frac{\bm P(k_n, m)}{p_{\max}} )      \\
      & =  \sum_{m \in M}\bm  b(m) \cdot \sum_{t=1}^{n} \prod_{j=1}^{t} \bm P(k_t, m).  \\
  \end{align*}

  Regarding the updated belief, according to the Bayesian update and the likelihood that the user will like all these $n$ categories, the updated belief will be:
  \[
    \forall m \in M: \bm b'(m) = \frac{\bm b(m) \cdot \prod_{t=1}^{n} \bm P(k_t, m)}{\sum_{m' \in M} \bm b(m') \cdot \prod_{t=1}^{n} \bm P(k_t, m')}.
  \]

  In the discussed POMDP model, for similar reasons, the updated belief will be:
  \[
    \forall m \in M: \bm b'(m) = \frac{\bm b(m) \cdot \bm P(k_1, m) \cdot \prod_{t=2}^{n} \frac{\bm P(k_t, m)}{p_{\max}}}{\sum_{m' \in M} \bm b(m') \cdot \bm P(k_1, m') \cdot \prod_{t=2}^{n} \frac{\bm P(k_t, m')}{p_{\max}}} = \frac{\bm b(m) \cdot \prod_{t=1}^{n} \bm P(k_t, m)}{\sum_{m' \in M} \bm b(m') \cdot \prod_{t=1}^{n} \bm P(k_t, m')}.
  \]

  In both models, when receiving "dislike" feedback, the user leaves the system, and there is a full certainty that the system transitioned into the absorbing state. This is not relevant for policy equivalence, as all of the policies guarantee the same reward of $0$ when they are in the absorbing state.

  After being convinced of the equivalency of the two models, we can use the POMDPs literature to our advantage.
  Specifically, we use the PBVI algorithm presented in \citet{point-based-value-iteration-an-anytime-algorithm-for-pomdps} to solve the POMDP we constructed in the desired complexity.
  Note that in our constructed POMDP the discount factor is less than 1, so the guarantees of the PBVI algorithm hold.

    
  \paragraph{Complexity} Let $B \subset \Delta(M)$ be an arbitrarily chosen finite subset of the belief simplex. While this discretization is arbitrary, its density—characterized by $\delta = \max_{b' \in \Delta(M)} \min_{b \in B} \|\bm{b} - \bm{b}'\|_1$—determines the approximation quality of our solution. Intuitively, $\delta$ measures how well $B$ covers the belief simplex by quantifying the maximal distance from any belief to its nearest neighbor in $B$.

  The algorithm exhibits a computational complexity of $O(\abs{B}^2 \cdot \abs{K} \cdot \abs{M})$ per iteration, attributable to the evaluation of belief-to-belief transitions under different actions. Across $H$ iterations, this yields a total complexity of $O(H \cdot \abs{B}^2 \cdot \abs{K} \cdot \abs{M})$.

  When $B$ is sufficiently dense—that is, when $\delta$ is adequately small—we can establish rigorous error bounds on our solution. Specifically, the policy returned by the algorithm is guaranteed to be $\frac{R_{\max} \delta}{(1 - \gamma)^2}$-close to the optimal policy of the finite horizon problem, where $R_{\max}$ represents the maximal immediate reward in the POMDP and $\gamma$ denotes the discount factor. For our particular POMDP formulation, selecting $\delta = \varepsilon \cdot (1 - p_{\max})^2$ ensures the algorithm's output remains within $\varepsilon$ of the optimal value.

  We use a well-known result in measure theory for our advantage - the covering number of the high-dimensional simplex is $O\left( \left( \frac{1}{\delta} \right)^{\abs{M} - 1} \right)$. Therefore, one can construct a $\delta$-discretization of the belief simplex by using $O\left( \left( \frac{1}{\delta} \right)^{\abs{M} - 1} \right)$ beliefs.
  
  All in all, the complexity of the algorithm is $O\left( H \cdot \left( \frac{1}{\varepsilon \cdot (1 - p_{\max})^2} \right)^{2\abs{M} - 2} \cdot \abs{K} \cdot \abs{M} \right)$.
This completes the proof of the proposition.
\end{proof}

\section{Full Proof and Details of Theorem \ref{thm:convergence}}
\label{sec:convergence-of-the-optimal-policy-proofs}
In this section, we present a complete proof of Theorem~\ref{thm:convergence}, which establishes the convergence of optimal policies after a finite number of rounds. The proof is structured in three parts. First, we introduce and analyze a quantity $c$ that captures the complexity of a given instance. Then, we prove the main theorem using several auxiliary results that characterize the behavior of belief walks and optimal policies. Finally, we provide detailed proof of these auxiliary claims. We conclude the section with an illustrative example showing how optimal belief walks can behave near vertices. Specifically, the example demonstrates why we need to partition vertices into two distinct groups - those that act as stable convergence points and those that can only be temporary stops in the belief walk.

\subsection{Introducing $c$}

Before proving Theorem~\ref{thm:convergence}, we introduce the quantity $c = c(\bm q, \bm P)$. To that end, we define five helpful quantities $c_1, \ldots, c_4$ and use them to define $c$.

\begin{enumerate}
    \item $c_1 = 1 - p_{\max}$ - reflects the maximal amount any user type can like any category and can be used to represent the highest expected social welfare in any belief on the simplex.
    \item $c_2 = \min_{k \in K, m, m' \in M, m \neq m'} \{ \abs{\bm{P}(k, m) - \bm{P}(k, m') } \}$ - can be thought of as a measure of heterogeneity between different user types, as it quantifies the difference in how different user types like the same category.
    \item $c_3 = \min_{m \in M} \{ \max_{k} \bm{P}(k, m)  - \max_{k' \neq k} \bm{P}(k', m) \}$ - quantifies the gap between the favorite and second favorite categories for any user type. It affects the complexity of the recommendation - if this gap is small, it might necessitate more exploration before transitioning to myopic recommendation.
    \item $c_4 = \min_{m \in M} \{ q(m) \}$ - corresponds to "rare" user types that the recommender might need to take their preferences into account.
\end{enumerate}

Finally, we aggregate the above variables into the variable $c$:
\[
c = \min_{i=1, ..., 4} c_i.
\]

For $c_4$ to be larger than $0$, each prior entry should be larger than $0$. This is quite logical - if a user type has zero probability of arriving, it should not affect either the policy or the reward, and we can solve the smaller problem induced by removing its corresponding entries from both $\bm P$ and $\bm q$. 

In contrast to the straightforward justification for positive $c_1$ and $c_4$, the necessity of positive $c_2$ and $c_3$ warrants deeper examination. The condition $c_2 > 0$ rules out situations where various user types display the same preferences across any category. Furthermore, $c_3 > 0$ guarantees that every user type has a distinctly favored category. Although theoretically possible to violate, our convergence outcome relies on these distinctions to facilitate significant belief updates and to ensure that strictly dominant recommendations ultimately arise. The technical implications of these conditions will become apparent in the subsequent proofs, where we demonstrate how they enable both informative Bayesian updates and the eventual transition to pure exploitation. We clearly highlight the necessity of these assumptions and address the implications when they are not present.

It is worth noting that any instance with distinct preferences (DP) - where each entry of the matrix $\bm P$ is distinct - necessarily has $c_2, c_3 \neq 0$. Therefore, if such an instance satisfies the aforementioned logical conditions regarding $c_1$ and $c_4$, then Theorem~\ref{thm:convergence} holds for it. This observation demonstrates the broad applicability of our result to practical scenarios where user preferences exhibit natural variation.

While we have focused on DP instances as a natural example, it is important to emphasize that Theorem~\ref{thm:convergence} holds for any instance where $c > 0$. The DP instances represent a significant subset of such cases, as their structural properties inherently ensure positive values for several components of $c$, but they are by no means exhaustive.

\subsection{Full proof}

\begin{proof}[\normalfont\bfseries Proof of Theorem~\ref{thm:convergence}]

Given a user type $m \in M$ and $\delta > 0$, we say a belief $\bm{b}$ is \emph{$(\delta,m)$-concentrated} if $\bm{b}(m) > 1-\delta$. Typically, the prior $\bm{q}$ is \emph{$c$-unconcentrated}, meaning $\bm{q}(m) < 1-c$ for all $m \in M$, with probability mass distributed across multiple types.

Theorem~\ref{thm:gap-between-value-function} guarantees a strict, positive gap in the optimal value function of two consecutive unconcentrated beliefs in the optimal belief walk. This allows us to derive Corollary~\ref{corr:limited-unconcetrated}, which states that the number of unconcentrated beliefs in the optimal belief walk is finite. It uses potential arguments and the fact that the value function is bounded from above - if there were infinitely many unconcentrated beliefs in the optimal belief walk, then because of the positive gap between them, we would get that the optimal belief walk is unbounded.

Corollary~\ref{corr:limited-unconcetrated} guarantees that the belief walk must eventually reach concentrated beliefs. Our main tool there, Theorem~\ref{thm:myopic-near-boundary} uses the liphshits property of the optimal value function to state that as long as the belief walk remains concentrated near the same type $m$ in the simplex, the optimal policy will continue to recommend the same category $k = \argmax_{k'} \bm P(k',m)$.

At first glance, it seems to be enough - we have shown that the belief walk reaches concentrated beliefs, and as long as it remains concentrated, the recommendation is fixed. We still need to address the following cases:
\begin{enumerate}
    \item \textbf{What if the belief walk returns to unconcentrated beliefs?} Although it might happen, Corollary~\ref{corr:limited-unconcetrated} bounds the number of times it can happen, and therefore, after finitely many returns to unconcentrated beliefs, the belief walk must remain concentrated.
    \item \textbf{What if the belief walk transitions from one concentrated subspace of beliefs that corresponds to a certain type to another subspace of concentrated beliefs?} In this case, Corollary~\ref{corr:limited-unconcetrated} is useless, and the optimal recommendation might change if the types exhibit different preferences. Fortunately, Lemma~\ref{lemma:concentrated-transition} states that this transition is impossible without encountering an unconcentrated belief in between. Now, Corollary~\ref{corr:limited-unconcetrated} is useful - it again bound the number of times this case might happen.
\end{enumerate}

To conclude, after finite $T$ many rounds, the belief walk must remain concentrated near the same type forever. Then, Theorem~\ref{thm:myopic-near-boundary} ensures that the optimal policy will remain fixed.
\end{proof}

\subsection{Proofs of the Auxiliary Claims Used for Proving Theorem~\ref{thm:convergence}}

\begin{proof}[\normalfont\bfseries Proof of Theorem~\ref{thm:gap-between-value-function}]
    Recall that from Observation~\ref{obs:recursive-formula-of-the-value-function} we get that:
    \[
        V^{\star}(\bm{b}) = p_{\pi^{\star}_1(\bm{b})}(\bm{b}) \cdot \left( 1 + V^{\star}(\tau(\bm{b}, \pi^{\star}_1(\bm{b}))) \right) \implies V^{\star}(\tau(\bm{b}, \pi^{\star}_1(\bm{b}))) = \frac{V^{\star}(\bm{b})}{p_{\pi^{\star}_1(\bm{b})}(\bm{b})} - 1.
    \]

    \textbf{Note:} As $\bm P$ is DP, only one entry in each row can be $0$, and since $\bm b$ is $\delta$-unconcentrated, at least two entries in $\bm b$ must be greater than $0$. Therefore, in the sum $p_{\pi^{\star}_1(\bm{b})}(\bm{b}) = \sum_{m \in M} \mathbf{P}(\pi^{\star}_1(\mathbf{b}), m) \cdot \mathbf{b}(m)$, there must be at least one product where both terms are positive. Consequently, $p_{\pi^{\star}_1(\bm{b})}(\bm{b}) > 0$.

    Define $\hat{\pi} = \left( \pi^{\star}_1(\bm{b}) \right)_{i=1}^{\infty}$, the policy that consistently recommends the item $\pi^{\star}_1(\bm{b})$.
    Then:
    \begin{align*}
       V^{\star}(\tau(\mathbf{b}, \pi^{\star}_1(\mathbf{b}))) - V^{\star}(\mathbf{b}) 
       &= \frac{V^{\star}(\mathbf{b})}{p_{\pi^{\star}_1(\mathbf{b})}(\mathbf{b})} - 1 - V^{\star}(\mathbf{b}) \\
       &= \frac{V^{\star}(\mathbf{b}) \left( 1 - p_{\pi^{\star}_1(\mathbf{b})}(\mathbf{b}) \right) - p_{\pi^{\star}_1(\mathbf{b})}(\mathbf{b})}{p_{\pi^{\star}_1(\mathbf{b})}(\mathbf{b})} \\
       &\geq_{(1)} \frac{V^{\hat{\pi}}(\mathbf{b}) \left( 1 - p_{\pi^{\star}_1(\mathbf{b})}(\mathbf{b}) \right) - p_{\pi^{\star}_1(\mathbf{b})}(\mathbf{b})}{p_{\pi^{\star}_1(\mathbf{b})}(\mathbf{b})} \\
       &= \frac{V^{\hat{\pi}}(\mathbf{b})}{p_{\pi^{\star}_1(\mathbf{b})}(\mathbf{b})} - 1 - V^{\hat{\pi}}(\mathbf{b}) \\
       &= V^{\hat{\pi}}(\tau(\mathbf{b}, \pi^{\star}_1(\mathbf{b}))) - V^{\hat{\pi}}(\mathbf{b}) \\
       &= p_{\pi^{\star}_1(\mathbf{b})}(\tau(\mathbf{b}, \pi^{\star}_1(\mathbf{b}))) \cdot \left( 1 + V^{\hat{\pi}}(\tau(\tau(\mathbf{b}, \pi^{\star}_1(\mathbf{b})), \pi^{\star}_1(\mathbf{b}))) \right) \\
       &\quad - p_{\pi^{\star}_1(\mathbf{b})}(\mathbf{b}) \cdot \left( 1 + V^{\hat{\pi}}(\tau(\mathbf{b}, \pi^{\star}_1(\mathbf{b}))) \right) \\
       &\geq_{(2)} \left( p_{\pi^{\star}_1(\mathbf{b})}(\tau(\mathbf{b}, \pi^{\star}_1(\mathbf{b}))) - p_{\pi^{\star}_1(\mathbf{b})}(\mathbf{b}) \right) \cdot \left( 1 + V^{\hat{\pi}}(\tau(\mathbf{b}, \pi^{\star}_1(\mathbf{b}))) \right) \\
       &\geq p_{\pi^{\star}_1(\mathbf{b})}(\tau(\mathbf{b}, \pi^{\star}_1(\mathbf{b}))) - p_{\pi^{\star}_1(\mathbf{b})}(\mathbf{b}),
    \end{align*}
    where:
    \begin{enumerate}
        \item $V^{\star}(\bm{b}) \geq V^{\hat{\pi}}(\bm{b})$ because $\pi^\star$ is optimal and $\left( 1 - p_{\pi^{\star}_1(\bm{b})}(\bm{b}) \right) \geq 0$.
        \item $V^{\hat{\pi}}(\tau(\tau(\bm{b}, \pi^{\star}_1(\bm{b})), \pi^{\star}_1(\bm{b}))) = \sum_{t=1}^{\infty} \prod_{j=1}^{t} p_{\pi^{\star}_1(\bm{b})}(\bm{b}^{\pi^{\star}, \bm{q}}_{j+2}) \geq \sum_{t=1}^{\infty} \prod_{j=1}^{t} p_{\pi^{\star}_1(\bm{b})}(\bm{b}^{\pi^{\star}, \bm{q}}_{j+1}) = V^{\hat{\pi}}(\tau(\bm{b}, \pi^{\star}_1(\bm{b})))$, where the inequality is due to $p_k(\tau(\bm{b}, k)) \geq p_k(\bm{b})$, which will be proven later.
    \end{enumerate}

    To summarize the above steps, we discovered that in order to bound the gap between two consecutive optimal values $V^\star(\bm b), V^\star(\tau(\bm b), \pi^\star_1(b))$ from below, it is enough to bound the expression $p_{\pi^{\star}_1(\bm{b})}(\tau(\bm{b}, \pi^{\star}_1(\bm{b}))) - p_{\pi^{\star}_1(\bm{b})}(\bm{b})$. Let's develop it using Bayesian updates:
    \begin{align*}
       p_{\pi^{\star}_1(\mathbf{b})}(\tau(\mathbf{b}, \pi^{\star}_1(\mathbf{b}))) - p_{\pi^{\star}_1(\mathbf{b})}(\mathbf{b}) 
       &= \sum_{m \in M} \mathbf{P}(\pi^{\star}_1(\mathbf{b}), m) \cdot \tau(\mathbf{b}, \pi^{\star}_1(\mathbf{b}))(m) - \sum_{m \in M} \mathbf{P}(\pi^{\star}_1(\mathbf{b}), m) \cdot \mathbf{b}(m) \\
       &= \sum_{m \in M} \mathbf{P}(\pi^{\star}_1(\mathbf{b}), m) \cdot \frac{\mathbf{P}(\pi^{\star}_1(\mathbf{b}), m) \cdot \mathbf{b}(m)}{\sum_{m' \in M} \mathbf{P}(\pi^{\star}_1(\mathbf{b}), m') \cdot \mathbf{b}(m')} \\
       &\quad - \sum_{m \in M} \mathbf{P}(\pi^{\star}_1(\mathbf{b}), m) \cdot \mathbf{b}(m) \\
       &= \frac{\sum_{m \in M} \mathbf{P}(\pi^{\star}_1(\mathbf{b}), m)^2 \cdot \mathbf{b}(m) - \left( \sum_{m \in M} \mathbf{P}(\pi^{\star}_1(\mathbf{b}), m) \cdot \mathbf{b}(m) \right)^2}{\sum_{m \in M} \mathbf{P}(\pi^{\star}_1(\mathbf{b}), m) \cdot \mathbf{b}(m)}.
    \end{align*}

    Define $X$ to be the following random variable:
    \[
        X = \begin{cases}
              P(\pi^{\star}_1(\bm{b}), m) & \text{with probability } b(m) \quad \forall m \in M . 
        \end{cases}
    \]
    Notice that the numerator of the expression above is the variance of $X$ and the denominator is the expectation of $X$. This proves the claim $p_k(\tau(\bm{b}, k)) \geq p_k(\bm{b})$ which we used earlier, as both of those expressions are non-negative.

    We want to bound away the expression $p_{\pi^{\star}_1(\bm{b})}(\tau(\bm{b}, \pi^{\star}_1(\bm{b}))) - p_{\pi^{\star}_1(\bm{b})}(\bm{b})$ from $0$, and for that, we will use the assumption on $\bm b$ - each entry in $b(m)$ is at most $1 - \delta$. This will help us to bound from below the variance of $X$, as this condition guarantees that $X$ is not all concentrated at a single point.
    The expectation of $X$ is bounded from above by $1 - c$.
    Furthermore, the variance of $X$ can be bounded from below by inserting the mock random variable $\hat{X}$:
    \[
        \hat{X} = \begin{cases}
            0         & \text{with probability } 1 - \delta. \\
            c & \text{with probability } \delta.    
        \end{cases}
    \]
    This will be a lower bound because the distribution of $X$ is more spread out than the distribution of $\hat{X}$ as each entry in $b(m)$ is at most $1 - \delta$, and we defined $c$ to be not bigger than $\min_{k \in K, m, m' \in M, m \neq m'} \{ \abs{\bm{P}(k, m) - \bm{P}(k, m') } \}$, and therefore the distance of the closest entry to the one with the highest mass must be at least of distance $c$.
    The variance of $\hat{X}$ is 
    \[
        \delta \cdot (c^2)- \left( \delta \cdot c \right)^2 = (1 - \delta) \cdot \delta \cdot c^2.
    \]

    In total we find that the expression $p_{\pi^{\star}_1(\bm{b})}(\tau(\bm{b}, \pi^{\star}_1(\bm{b}))) - p_{\pi^{\star}_1(\bm{b})}(\bm{b})$ is at least:
    \[
        \frac{(1 - \delta) \cdot \delta \cdot c^2}{1 - c};
    \]
therefore,
    \[
        V^{\star}(\tau(\bm{b}, \pi^{\star}_1(\bm{b}))) - V^{\star}(\bm{b}) \geq \frac{\delta \cdot (1 - \delta) \cdot c^2}{1 - c}.
    \]
\end{proof}

\paragraph{Note about the definition of $c$} The inequality $0 < c \leq c_2$ plays a crucial role in our analysis, particularly in bounding the variance of the random variable $X$ in Theorem~\ref{thm:gap-between-value-function}. The necessity of maintaining $c_2 > 0$ extends beyond mere technical convenience; it represents a fundamental requirement for the theorem's validity. This condition ensures that each action elicits distinctly differentiated responses across user types, thereby facilitating meaningful Bayesian updates and monotonic growth in potential value. To illustrate the indispensability of this condition, consider a counterexample where an action's corresponding row in matrix $\bm P$ contains identical entries. In such a case, recommendations of this category would fail to modify the prior belief, effectively nullifying the monotonicity property central to our analysis. 

\begin{proof}[\normalfont\bfseries Proof of Corollary~\ref{corr:limited-unconcetrated}]
    Consider the belief walk $(\bm{b}^{\pi^{\star}, \bm{q}}_t)_{t=1}^{\infty}$ induced by the optimal policy $\pi^{\star}$ starting from the prior $\bm{q}$. Let $S$ be the sequence of indices $t$ where $\bm{b}^{\pi^{\star}, \bm{q}}_t$ is $\delta$-unconcentrated. Assume for contradiction that $|S| > H$, and denote $r$ to be the $H+1$'th index in $S$. 
    
    By Theorem~\ref{thm:gap-between-value-function}, each unconcentrated belief in the walk contributes an increase of at least $\frac{\delta \cdot (1 - \delta) \cdot c^2}{(1 - c)}$ to the optimal value function. Therefore, having more than $H$ unconcentrated beliefs would imply that $V^{\star}(\bm{b}^{\pi^{\star}, \bm{q}}_{r+1}) > \frac{1-c}{c}$, contradicting the upper bound on the value function established in Equation~\eqref{eq:bounded-value}.
\end{proof}

\begin{lemma}
    \label{lemma:lipschitz}
    The optimal value function $V^\star$ is Lipschitz continuous with constant $L = \frac{1-c}{c}$.
\end{lemma}

\begin{proof}[\normalfont\bfseries Proof of Lemma~\ref{lemma:lipschitz}]
Let $\bm b, \bm b' \in \Delta(M)$ be two beliefs. Using $|a-b| = max\{a-b, b-a\}$, we have:
\begin{align*}
|V^{\star}(\mathbf{b}) - V^{\star}(\mathbf{b}')| &= \max\{V^{\star}(\mathbf{b}) - V^{\star}(\mathbf{b}'), V^{\star}(\mathbf{b}') - V^{\star}(\mathbf{b})\} \\
&\leq_{(1)} \max\{V^{\pi^{\star}(\mathbf{b})}(\mathbf{b}) - V^{\pi^{\star}(\mathbf{b})}(\mathbf{b}'), V^{\pi^{\star}(\mathbf{b}')}(\mathbf{b}') - V^{\pi^{\star}(\mathbf{b}')}(\mathbf{b})\} \\
&= \max\{\sum_{m \in M} (\mathbf{b}(m) - \mathbf{b}'(m)) \cdot \left( \sum_{t=1}^{\infty} \prod_{i=1}^{t} \mathbf{P}(\pi^\star_i(\mathbf{b}), m) \right), \\
&\quad\quad\quad \sum_{m \in M} (\mathbf{b}'(m) - \mathbf{b}(m)) \cdot \left( \sum_{t=1}^{\infty} \prod_{i=1}^{t} \mathbf{P}(\pi^\star_i(\mathbf{b}'), m) \right)\} \\
&\leq \sum_{m \in M} |\mathbf{b}(m) - \mathbf{b}'(m)| \cdot \left( \sum_{t=1}^{\infty} \prod_{i=1}^{t} p_{\max} \right) \\
&= |\mathbf{b} - \mathbf{b}'|1 \cdot \frac{p_{\max}}{1-p_{\max}} \\
&\leq_{(2)} |\mathbf{b} - \mathbf{b}'|1 \cdot \frac{1-c}{c},
\end{align*}
where:
\begin{enumerate}
    \item $\pi^{\star}(\mathbf{b})$ is optimal for $\bm b$, $\pi^{\star}(\mathbf{b'})$ is optimal for $\bm b'$.
    \item $p{\max} \leq 1 - c \implies \frac{p_{\max}}{1-p_{\max}} \leq \frac{1-c}{c}$.
\end{enumerate}
\end{proof}

\begin{proof}[\normalfont\bfseries Proof of Theorem~\ref{thm:myopic-near-boundary}]
    Let $m \in M$ be a user type, $k = \argmax_{k' \in K} \bm{P}(k', m)$ and $\bm{b} \in \Delta(M)$ such that $\bm b(m) \geq 1 - \frac{c^2}{4}$.

    Denote $\pi^{k} = \left( k \right)_{i=1}^{\infty}$, the policy that always recommends category $k$. Additionally, let some arbitrary policy $\hat{\pi}$ such that $\hat{\pi}(\bm{b})_1 = \hat{k}$ for some $\hat{k} \neq k$. We don't assume anything else about $\hat{\pi}$. We will show that $V^{\pi^k}(\bm{b}) \geq V^{\hat{\pi}}(\bm{b})$ and this will imply that the first recommendation of the optimal policy is $k$.

    Notice that:
    \begin{itemize}
        \item $V^{\pi^k}(\bm{b}) = \sum_{m' \in M} b(m') \cdot \frac{\bm{P}(k, m')}{1 - \bm{P}(k, m')}$ \hfill (due to Lemma~\ref{lemma:closed-form-representations-of-the-value-function}).
        \item $V^{\hat{\pi}}(\bm{b}) = \left( \bm b(m) \cdot \bm{P}(\hat{k}, m) + \sum_{m' \neq m} \bm b(m') \cdot \bm{P}(\hat{k}, m') \right) \cdot \left( 1 + V^{\hat{\pi}[2:]}(\tau(\bm{b}, \hat{k})) \right)$ \hfill (due to Observation~\ref{obs:recursive-formula-of-the-value-function}).
    \end{itemize}
    We will now analyze the relation between $V^{\pi^k}_{L}(\bm{b})$ and $V^{\hat{\pi}}_{U}(\bm{b})$, where $V^{\pi^k}_{L}(\bm{b})$ is a lower bound on $V^{\pi^k}(\bm{b})$ and $V^{\hat{\pi}}_{U}(\bm{b})$ is an upper bound on $V^{\hat{\pi}}(\bm{b})$. We will soon see that $V^{\pi^k}_{L}(\bm{b}) \geq V^{\hat{\pi}}_{U}(\bm{b})$ and this will conclude the proof.

    We define
    \[
        V^{\pi^k}_{L}(\bm{b}) := \left( 1 - \frac{c^2}{4} \right) \cdot \frac{\bm{P}(k, m)}{1 - \bm{P}(k, m)} \overset{\bm b(m) \geq 1 - \frac{c^2}{4}}{\leq} \bm b(m) \cdot \frac{\bm{P}(k, m)}{1 - \bm{P}(k, m)} \leq \sum_{m' \in M} b(m') \cdot \frac{\bm{P}(k, m')}{1 - \bm{P}(k, m')} = V^{\pi^k}(\bm{b}).
    \]

    For constructing the upper bound $V^{\hat{\pi}}_{U}(\bm{b})$, we first note that according to Lemma~\ref{lemma:lipschitz} it holds that:
    \[
        V^{\hat{\pi}[2:]}(\tau(\bm{b}, \hat{k})) \leq V^{\star}(\tau(\bm{b}, \hat{k})) \overset{Lemma~\ref{lemma:lipschitz}}{\leq} V^{\star}(\bm e_m) + \norm{\tau(\bm{b}, \hat{k}) - \bm e_m}_1 \cdot \frac{1-c}{c}.
    \]
    Therefore we can define the upper bound as:
    \begin{align*}
       V^{\hat{\pi}}_{U}(\mathbf{b}) &:= \left( \mathbf{b}(m) \cdot \mathbf{P}(\hat{k}, m) + \sum_{m' \neq m} \mathbf{b}(m') \cdot \mathbf{P}(\hat{k}, m') \right) \cdot \left( 1 + V^\star(\mathbf{e}_m) + \|\tau(\mathbf{b}, \hat{k}) - \mathbf{e}_m\|_1 \cdot \frac{1-c}{c} \right) \\
       &= \left( \mathbf{b}(m) \cdot \mathbf{P}(\hat{k}, m) + \sum_{m' \neq m} \mathbf{b}(m') \cdot \mathbf{P}(\hat{k}, m') \right)
       \cdot \left( \frac{1}{1 - \mathbf{P}(k, m)} + \|\tau(\mathbf{b}, \hat{k}) - \mathbf{e}_m\|_1 \cdot \frac{1-c}{c} \right).
    \end{align*}
    Where the equality is due to $V^{\star}(\bm e_m) = \frac{\bm{P}(k, m)}{1 - \bm{P}(k, m)}$ since $\bm e_m$ is a deterministic belief and therefore recommending the category that is the most liked by the user type $m$ is optimal.    

    The expression $\norm{\tau(\bm{b}, \hat{k}) - \bm e_m}_1$ can be simplified as follows:
    \begin{align*}
       \|\tau(\mathbf{b}, \hat{k}) - \mathbf{e}_m\|_1 
       &= \sum_{m' \in M} |\tau(\mathbf{b}, \hat{k})(m') - \mathbf{e}_m(m')| \\
       &= \left( 1 - \frac{\mathbf{b}(m) \cdot \mathbf{P}(\hat{k}, m)}{\mathbf{b}(m) \cdot \mathbf{P}(\hat{k}, m) + \sum_{m' \neq m} \mathbf{b}(m') \cdot \mathbf{P}(\hat{k}, m')} \right) \\
       &+ \frac{\sum_{m' \neq m} \mathbf{b}(m') \cdot \mathbf{P}(\hat{k}, m')}{\mathbf{b}(m) \cdot \mathbf{P}(\hat{k}, m) + \sum_{m' \neq m} \mathbf{b}(m') \cdot \mathbf{P}(\hat{k}, m')} \\
       &= \frac{2 \cdot \sum_{m' \neq m} \mathbf{b}(m') \cdot \mathbf{P}(\hat{k}, m')}{\mathbf{b}(m) \cdot \mathbf{P}(\hat{k}, m) + \sum_{m' \neq m} \mathbf{b}(m') \cdot \mathbf{P}(\hat{k}, m')}.
    \end{align*}
    Plugging $\norm{\tau(\bm{b}, \hat{k}) - \bm e_m}_1$ into $V^{\hat{\pi}}_{U}(\bm{b})$ yields:
    \begin{align*}
         V^{\hat{\pi}}_{U}(\bm{b}) &= \left( \bm b(m) \cdot \bm{P}(\hat{k}, m) + \sum_{m' \neq m} \bm b(m') \cdot \bm{P}(\hat{k}, m') \right) \\
         &\cdot \left( \frac{1}{1 - \bm{P}(k, m)} + \frac{2 \cdot \sum_{m' \neq m} \bm b(m') \cdot \bm{P}(\hat{k}, m')}{\bm b(m) \cdot \bm{P}(\hat{k}, m) + \sum_{m' \neq m} \bm b(m') \cdot \bm{P}(\hat{k}, m')} \cdot \frac{1-c}{c} \right) \\
         &= \left( \bm b(m) \cdot \bm{P}(\hat{k}, m) + \sum_{m' \neq m} \bm b(m') \cdot \bm{P}(\hat{k}, m') \right) \cdot \left( \frac{1}{1 - \bm{P}(k, m)} \right) \\
         &+ \left( 2 \cdot \sum_{m' \neq m} \bm b(m') \cdot \bm{P}(\hat{k}, m') \right) \cdot \frac{1-c}{c}. \\ 
    \end{align*}
    We again use the fact that $\bm b(m) \geq 1 - \frac{c^2}{4}$ to further bound the expression from above:
    \begin{align*}
        V^{\hat{\pi}}_{U}(\bm{b}) &= \left( \bm b(m) \cdot \bm{P}(\hat{k}, m) + \sum_{m' \neq m} \bm b(m') \cdot \bm{P}(\hat{k}, m') \right) \cdot \left( \frac{1}{1 - \bm{P}(k, m)} \right) \\
        &+ \left( 2 \cdot \sum_{m' \neq m} \bm b(m') \cdot \bm{P}(\hat{k}, m') \right) \cdot \frac{1-c}{c} \\
        &\leq_{(1)} \left( \bm b(m) \cdot \bm{P}(\hat{k}, m) + \frac{c^2}{4} \right) \cdot \left( \frac{1}{1 - \bm{P}(k, m)} \right) + \left( 2 \cdot \frac{c^2}{4} \right) \cdot \frac{1-c}{c} \\
        &\leq_{(2)} \left(\bm{P}(k, m) - c + \frac{c^2}{4} \right) \cdot \left( \frac{1}{1 - \bm{P}(k, m)} \right) + \left( 2 \cdot \frac{c^2}{4} \right) \cdot \frac{1-c}{c} \\
        &= \left(\bm{P}(k, m) - c + \frac{c^2}{4} \right) \cdot \left( \frac{1}{1 - \bm{P}(k, m)} \right) + \frac{c}{2} - \frac{c^2}{2} \\
        &\leq \frac{\bm{P}(k, m) - c + \frac{c^2}{4} + \frac{c (1- \bm P(k, m))}{2}}{1 - \bm{P}(k, m)} \\
        &\leq \frac{\bm{P}(k, m) - c + \frac{c^2}{4} + \frac{c}{2}}{1 - \bm{P}(k, m)} \\
        &= \frac{\bm{P}(k, m) - \frac{c}{2} + \frac{c^2}{4}}{1 - \bm{P}(k, m)} \\
        &\leq \frac{\bm{P}(k, m) - \frac{c}{4}}{1 - \bm{P}(k, m)} \\
        &\leq \frac{\bm{P}(k, m) - \bm{P}(k, m) \frac{c^2}{4}}{1 - \bm{P}(k, m)} \\
        &= \frac{\bm{P}(k, m)}{1 - \bm{P}(k, m)} \cdot \left( 1 - \frac{c^2}{4} \right), \\
    \end{align*}
    Where:
    \begin{enumerate}
        \item $\bm b(m) \geq 1 - \frac{c^2}{4} \implies 1 - \bm b(m) \leq \frac{c^2}{4}$ and $\sum_{m' \neq m} \bm b(m') \cdot \bm{P}(\hat{k}, m') \leq \sum_{m' \neq m} \bm b(m') = 1-\bm b(m) \leq \frac{c^2}{4}$
        \item Since $k = \argmax_{k' \in K} \bm P(k, m)$ we have that $\bm b(m) \cdot \bm{P}(\hat{k}, m) \leq \bm{P}(\hat{k}, m) \leq \bm{P}(k, m) -c$ from the definition of $c$.
    \end{enumerate}
    To summarize, we got that 
    \[
    V^{\pi^k}(\bm b) \geq V^{\pi^k}_{L}(\bm{b}) = \left( 1 - \frac{c^2}{4} \right) \cdot \frac{\bm{P}(k, m)}{1 - \bm{P}(k, m)} \geq V^{\hat{\pi}}_{U}(\bm{b}) \geq V^{\hat{\pi}}(\bm{b}).
    \]
    Therefore, the first category the optimal policy will recommend is $k$, as all of the policies that start with $\hat k \neq k$ are found to be sub-optimal.
\end{proof}

\paragraph{Note about $c_3$ and myopic behavior} The quantity $c_3$ assumes a pivotal role in ensuring the validity of Theorem~\ref{thm:myopic-near-boundary}, particularly in guaranteeing the emergence of a unique optimal recommendation near vertices. The condition $c_3 > 0$ is not merely technical; it ensures that preferences across user types maintain sufficient separation to enable definitive category selection. To understand its necessity, consider what occurs when $c_3 = 0$: this would allow for the existence of a user type whose preferences are identical across multiple categories. In such a scenario, even as our belief concentrates arbitrarily close to the corresponding vertex, we could not distinguish between these equally preferred categories, undermining the theorem's conclusion about convergence to a single optimal action. 

\begin{proof}[\normalfont\bfseries Proof of Lemma~\ref{lemma:concentrated-transition}]
Let $m, m' \in M$ be distinct user types and let $\bm{b}$ be a $(\frac{c^2}{4},m)$-concentrated belief. Let $\bm{b}' = \tau(\bm{b}, \pi^\star_1(\bm b))$ denote the updated belief after recommending the optimal recommendation $\pi^\star_1(\bm b)$. We will show that $\bm{b}'$ cannot be $(\frac{c^2}{4},m')$-concentrated.

By the definition of Bayesian updates, for any user type $\hat{m} \in M$:
\[
   \bm{b}'(\hat{m}) = \tau(\bm{b}, \pi^\star_1(\bm b))(\hat{m}) = \frac{\bm{b}(\hat{m}) \cdot \bm{P}(\pi^\star_1(\bm b), \hat{m})}{\sum_{m'' \in M} \bm{b}(m'') \cdot \bm{P}(\pi^\star_1(\bm b), m'')}.
\]

For the denominator, we can utilize the fact that $\bm{b}$ is $(\frac{c^2}{4},m)$-concentrated:
\[
   \sum_{m'' \in M} \bm{b}(m'') \cdot \bm{P}(\pi^\star_1(\bm b), m'') \geq \bm{b}(m) \cdot \bm{P}(\pi^\star_1(\bm b), m) \geq \left(1 - \frac{c^2}{4}\right) \cdot \bm{P}(\pi^\star_1(\bm b), m).
\]

As $\pi^\star_1(\bm b)$ is the optimal recommendation for $\bm b$, from the definition of $c$ we must have $\bm{P}(\pi^\star_1(\bm b), m') \geq c$. For $m'$, we know that $\bm{b}(m') < \frac{c^2}{4}$ (since $\bm{b}(m) \geq 1 - \frac{c^2}{4}$ and $m \neq m'$). Therefore:
\[
   \bm{b}'(m') = \frac{\bm{b}(m') \cdot \bm{P}(\pi^\star_1(\bm b), m')}{\sum_{m'' \in M} \bm{b}(m'') \cdot \bm{P}(\pi^\star_1(\bm b), m'')} \leq \frac{\frac{c^2}{4} \cdot 1}{\left(1 - \frac{c^2}{4}\right) \cdot c} = \frac{c}{4(1 - \frac{c^2}{4})} = \frac{c}{4 - c^2}.
\]

Now we can use the fact that the function $\frac{c}{4 - c^2}$ is strictly increasing in $[0,1]$ to get that for any $c \in [0,1]$, $\frac{c}{4 - c^2} \leq \frac{1}{3}$. Conversely, the strict decrease of $1 - \frac{c^2}{4}$ in $[0,1]$ implies that for any $c \in [0,1]$, $1 - \frac{c^2}{4} \geq \frac{3}{4}$. Hence, we conclude that $\bm{b}'(m') < 1 - \frac{c^2}{4}$, completing our proof.
\end{proof}

\subsection{Example of moving away from a vertex after starting close to it}

Earlier in this section, we proved Theorem~\ref{thm:myopic-near-boundary}, which argues that when close enough to any vertex of the belief simplex, the optimal single recommendation is the category that is the most liked by the user type that corresponds to the vertex. However, it does not mean that the optimal policy will always recommend the same category for a belief close to a vertex.
Consider the following matrix:
\[
    \bm{P} = \begin{pmatrix}
        0.8 & 0.5 \\
        0.7 & 0.6
    \end{pmatrix}.
\]

For a belief that is close enough to the second user type (corresponding to the second column), the optimal policy will recommend the second category, as this is the preferred category for this user type ($0.6 > 0.5$).
However, notice that the first column is pairwise bigger than the second column.
This means that no matter what category the recommender chooses, the belief will always be updated to be closer to the first user type.
Therefore, after a few steps, the belief will be close enough to the first user type, and the optimal policy will recommend the first category, as this is the preferred category for this user type ($0.8 > 0.7$).

\section{Omitted Proofs from Section \ref{sec:branch-and-bound-algorithm}}

\label{sec:branch-and-bound-algorithm-proofs}

\begin{proof}[\normalfont\bfseries Proof of Lemma~\ref{lemma:upper-lower-bound}]
    Notice that according to Lemma~\ref{lemma:closed-form-representations-of-the-value-function}, we have $V^L(\bm b) = V^{\pi_k}(\bm b)$ for $\pi_k = \left( k \right)_{t=1}^{\infty}$, where $k = \argmax_{k' \in K} \sum_{m \in M} \bm{b}(m) \cdot \frac{\bm{P}(k, m)}{1 - \bm{P}(k, m)}$ . Therefore, as $V^\star(\bm b) = V^{\pi^\star(b)}(\bm b)$ and  $\pi^\star(b)$ is optimal for $\bm b$ we have that $V^L(\bm b)$ is indeed a lower bound.

    As for the upper bound, notice that we can bound from above the value of $V^\star(\bm b)$ when using Lemma~\ref{lemma:closed-form-representations-of-the-value-function}.
    {
    \small
    \[
        V^\star(\bm{b}) = \sum_{m \in M} \bm{b}(m) \cdot \sum_{t=1}^{\infty} \prod_{j=1}^{t} \bm{P}(\pi^\star_j, m) \leq \sum_{m \in M} \bm{b}(m) \cdot \sum_{t=1}^{\infty} \prod_{j=1}^{t} \left( \max_{k} \bm{P}(k, m) \right) = \sum_{m \in M} \bm{b}(m) \cdot \max_{k} \frac{\bm{P}(k, m)}{1 - \bm{P}(k, m)} = V^{U}(\bm{b}).
    \]
    }
\end{proof}

\begin{proof}[\normalfont\bfseries Proof of Theorem~\ref{thm:bb-algorithm-bounded-error}]
    First of all, notice that after at most $H(\varepsilon) = \left\lceil \log_{p_{\max}} \frac{\varepsilon (1 - p_{\max})}{p_{\max}} \right\rceil$ steps in the belief walk, for every prefix $\Pi$ of length $H$ it holds that $\ovv_{\Pi} - \unv_{\Pi} \leq \varepsilon$, because:
    \begin{align*}
        \ovv - \unv &= \sum_{t=1}^{H} \prod_{i=1}^{t} p_{\pi_i}(\bm{b}^{\pi, \bm{q}}_i) + \prod_{i=1}^{H} p_{\pi_i}(\bm{b}^{\pi, \bm{q}}_{i}) \cdot V^{U}(\bm{b}^{\pi, \bm{q}}_{H+1}) - \sum_{t=1}^{H} \prod_{i=1}^{t} p_{\pi_i}(\bm{b}^{\pi, \bm{q}}_i) - \prod_{i=1}^{H} p_{\pi_i}(\bm{b}^{\pi, \bm{q}}_{i}) \cdot V^{L}(\bm{b}^{\pi, \bm{q}}_{H+1}) \\
        &= \prod_{i=1}^{H} p_{\pi_i}(\bm{b}^{\pi, \bm{q}}_{i}) \cdot V^{U}(\bm{b}^{\pi, \bm{q}}_{H+1}) - \prod_{i=1}^{H} p_{\pi_i}(\bm{b}^{\pi, \bm{q}}_{i}) \cdot V^{L}(\bm{b}^{\pi, \bm{q}}_{H+1}) \\
        &\leq \prod_{i=1}^{H} p_{\pi_i}(\bm{b}^{\pi, \bm{q}}_{i}) \cdot V^{U}(\bm{b}^{\pi, \bm{q}}_{H+1}) \\
        &\leq p_{\max}^H \cdot \frac{p_{\max}}{1-p_{\max}} \leq \varepsilon.
    \end{align*}

    Furthermore, notice every time the algorithm enters Line~\ref{bnbalg:branching} with a prefix $\Pi$ it holds that $\unv_{\Pi} \leq \tilde V$, as $\tilde V$ was updated earlier. Therefore, starting from depth $H+1$, no prefix will satisfy the condition and the algorithm will terminate.
    
    Assume in contradiction that Algorithm~\ref{bb-algorithm} returns a value $\tilde V$ such that $V^{\star}(\bm{q}) - \tilde V > \varepsilon$.

    Notice that the returned value of $\tilde V$ is the largest value that this variable has ever achieved during the run of the algorithm. Denote the last prefix of $\pi^\star$ that was pruned in the Bounding part of the algorithm as $\Pi'$, and the value of $\tilde V$ that bounded it as $V'$. It holds that $\ovv_{\Pi'} < V' + \varepsilon$, as $\Pi'$ was bounded. Because we have $V' \leq \tilde V$, it holds that $V^{\star}(\bm{q}) \leq \overline{V_{\Pi'}} < \tilde V + \varepsilon$, which contradicts the assumption.
        
\end{proof}


\section{Auxiliary Details about the Experiments}
\label{sec:auxiliary-details-about-the-experiments}

In this appendix, we provide technical details about the simulations in Section~\ref{sec:experiments}.

\subsection{Instance Generation}

Problem instances for the simulations were generated using a random sampling procedure designed to reflect realistic distributions of user preferences and categories. The procedure involves generating two key components: the matrix $\bm{P}$, representing the probabilities that each user type prefers each category, and the distribution $\bm{q}$, representing the prior belief over user types.

The matrix $\bm{P}$ is constructed by first independently sampling latent vectors for each user type and category from the standard normal distribution. These latent vectors are then normalized, and the cosine similarity between the vectors is computed, representing the affinity between each user type and category. This similarity is scaled to fall within the range $[0, 1]$, producing a probability matrix $\bm{P}$. To ensure numerical stability and avoid extreme values, the probabilities are clipped at $[0.01, 0.99]$.

The prior distribution $\bm{q}$ over user types is generated by independently sampling logits for each user type from a normal distribution $\mathcal{N}(0, 0.5)$. These logits are then transformed into probabilities using the softmax function, producing a categorical distribution over the user types. Similar to the probability matrix, the resulting distribution is clipped and normalized to maintain numerical stability and avoid extreme values.

Below are the Python functions used to implement this sampling procedure:

\begin{verbatim}
def get_random_P(n_actions, n_types, threshold=1e-2):
    """
    Generates the probability matrix P representing the likelihood 
    that each user type prefers each category.
    Latent vectors are sampled from a normal distribution, and 
    cosine similarities are computed and scaled to probabilities.
    """
    dim = n_actions
    latent_action_vectors = np.random.normal(0, 1, (n_actions, dim))
    latent_type_vectors = np.random.normal(0, 1, (n_types, dim))
    norm_a = np.linalg.norm(latent_action_vectors, ord=2, axis=1).reshape(-1, 1)
    norm_t = np.linalg.norm(latent_type_vectors, ord=2, axis=1).reshape(1, -1)
    norm = norm_a @ norm_t
    P = (latent_action_vectors @ latent_type_vectors.T) / norm
    P = (P + 1) / 2

    P = np.clip(P, threshold, 1-threshold)
    
    return P

def get_random_q(n_types, std=.5, threshold=1e-2):
    """
    Generates the prior distribution q over user types.
    Logits are sampled from a normal distribution and transformed 
    into probabilities using the softmax function.
    """
    log_q = np.random.normal(0, std, (n_types,))
    q = np.exp(log_q)
    q /= q.sum()
    q = np.clip(q, threshold/n_types, 1-threshold)
    q /= q.sum()
    return q
\end{verbatim}

\subsection{Statistical tests}

We verify the statistical significance of the difference between Algorithm~\ref{bb-algorithm} and the baseline algorithm reported in Figure \ref{fig:sarsop} using the Wilcoxon signed-rank test. For all data points (pairs of number of categories, number of user types), the two-sided test returns very small p-values -- less than $10^{-10}$ -- which is expected for sample sizes as large as 500.

\subsection{Hardware and Performance}

All experiments were conducted on a desktop PC equipped with 16 GB RAM and an 11th Gen Intel(R) Core i5-11600KF @3.90GHz processor. The simulations were run entirely on the CPU, and no GPU was utilized. It takes about two hours to run all experiments.

\section{MovieLens Experiments}\label{sec:movielens}

This section demonstrates an end-to-end procedure for applying our Branch and Bound algorithm to real-world recommendation data. Specifically, we show how to extract model parameters from the MovieLens 1M dataset, a widely used benchmark for collaborative filtering algorithms, and evaluate the performance of our approach against the baseline. This demonstrates that our method not only works theoretically but can be effectively applied to actual recommendation system data, providing insights into its practical efficiency and scalability.

\subsection{Dataset Description}
We evaluate our Branch and Bound approach using the MovieLens 1M dataset, containing $1$ a million ratings from $6,040$ users on $3,706$ movies. Each rating is an integer value between $1$ and $5$ stars, representing a realistic recommendation system scenario with typical characteristics such as sparsity and varying user activity levels.

\subsection{Clustering Methodology}

As part of our end-to-end pipeline for extracting model parameters, we first need to identify user types and item categories from raw rating data. This step is crucial for obtaining the inputs required by our Branch and Bound algorithm: the probabilities-to-like matrix and the prior distribution over user types.

We employ Spectral Co-clustering \citep{coclustering}, a well-established method for simultaneously clustering users and items in recommendation matrices \citep{george2005scalable}. This approach captures the dual nature of user-item interactions, identifying subgroups with distinct rating patterns that can serve as our user types and item categories.

For larger systems that might experience an imbalance between number of user types and item categories, we suggest applying additional clustering methods (e.g., k-means or hierarchical clustering) to decompose larger clusters. Since spectral co-clustering returns the same number of clusters on both domains, this post-processing step helps maintain balanced representation while preserving the computational advantages of our approach.

\subsection{Score Transformation}

The second step in our parameter extraction pipeline involves converting the clustered rating matrix into the probability matrix required by our model. For each user-type and item-category pair (identified through clustering), we calculate the mean rating and normalize it to the $[0,1]$ interval by dividing by the maximum rating ($5.0$). This transforms the raw ratings into probabilities that represent the likelihood of a user type "liking" a particular content category.

While this is a straightforward approach, alternative transformations can be employed, such as using percentile rankings or more sophisticated normalization techniques. Our choice prioritizes simplicity and interpretability while maintaining the relative preferences between clusters. The resulting probability matrix, combined with the prior distribution derived from cluster sizes, provides a complete instance of our model derived entirely from real-world data.

\subsection{Experimental Setup and Results}

Using the parameters extracted from the MovieLens dataset through our pipeline, we compare the Branch and Bound algorithm with SARSOP across different matrix sizes. This evaluation uses real-world derived parameters rather than synthetic ones, providing insight into the algorithms' performance in practical scenarios. The initial state distribution (prior) is derived from the proportional size of user clusters, reflecting the actual distribution of user types in the system.

To assess robustness to uncertainty in preference estimation, we add Gaussian noise ($\sigma = 0.01$) to the extracted probability matrices. We test both algorithms on $10$ random noise samples for each matrix size, and run each input $1000$ times to account for variations in runtime due to internal computational randomness. The resultant 95\% bootstrap confidence intervals are plotted as shaded regions around the mean curves, though they appear notably narrow due to the high consistency of runtime performance across experimental runs.

This evaluation with real-world derived parameters demonstrates that both algorithms achieve optimal solutions with markedly distinct computational characteristics. As illustrated in Figure~\ref{fig:movielens}, the Branch and Bound algorithm consistently outperforms SARSOP across all matrix dimensions, with the performance differential widening substantially as the problem size increases. This systematic superiority in computational efficiency, coupled with maintained solution quality, strongly indicates the algorithm's particular suitability for large-scale real-world recommendation systems. 

\begin{figure}
    \centering
    \includegraphics[width=0.5\linewidth]{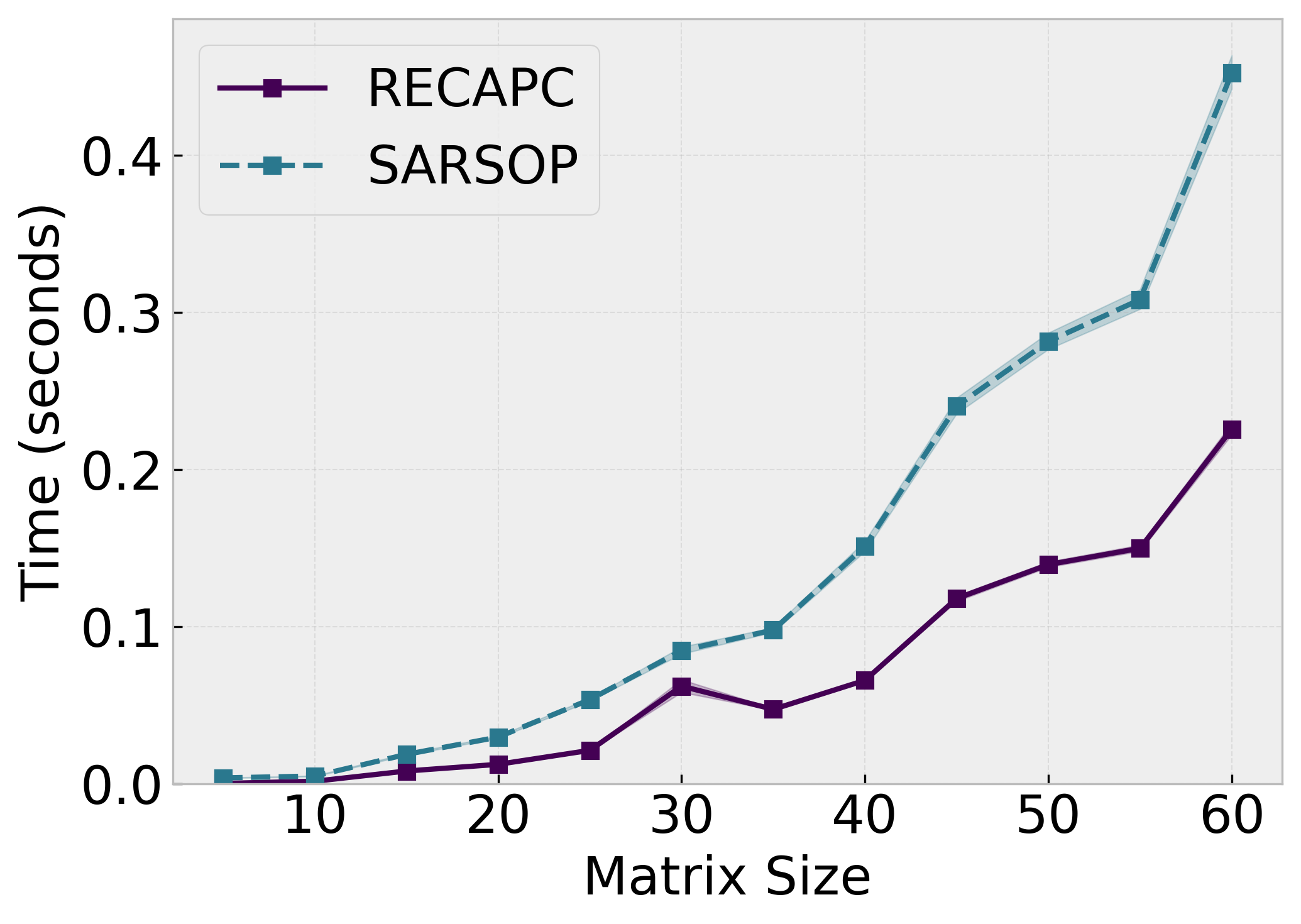}
    \caption{Comparative runtime analysis of Branch and Bound versus SARSOP algorithms on MovieLens dataset. Runtime (in seconds) is plotted against matrix dimension, with shaded regions representing 95\% bootstrap confidence intervals (indistinguishable due to high consistency across runs).}
    \label{fig:movielens}
\end{figure}

\section{Model Extensions and POMDP Framework}
\label{sec:model-extensions}

Our model provides a foundational framework for understanding recommendation systems under aggregated user information and churn risk. However, its binary interaction model -- where users either like and stay or dislike and leave -- represents a simplification of real-world user behavior. In this section, we discuss how our model can be extended within the broader POMDP framework to capture more nuanced user interactions.

\subsection{Richer User Interactions}

The current model can be extended to accommodate more complex user behaviors and reward structures. For instance, users might dislike a recommendation but remain in the system, potentially with a modified engagement level. Similarly, users might express varying degrees of satisfaction, from ``superlike'' to mild approval, each contributing differently to the system's reward. These extensions can be naturally formulated within the POMDP framework while maintaining our state space (user types) and action space (recommendation categories). The key enrichment comes through more complex transition, observation, and reward models. The transition function would capture various user responses, including remaining in the system after negative feedback and possibly changing preferences over time; the observation model would accommodate different levels of feedback (e.g., ratings on a scale), and the reward function would assign different values to these varied interaction types. This richer probabilistic structure better reflects real-world user behavior while preserving the fundamental representation of states as user types.

\subsection{Learning from Historical Data}

While our current model assumes known parameters ($\mathbf{P}$ and $\mathbf{q}$), a more practical approach would involve estimating these parameters along with the transition, observation, and reward models from historical data before deployment. The POMDP framework readily accommodates such model estimation - historical interaction data can approximate how users transition between states, how different types of feedback relate to underlying user types, and what rewards are obtained from various interactions. The quality of these empirical estimates directly impacts the effectiveness of the resulting recommendation policy obtained through POMDP solvers. While this introduces additional complexity in model specification, it allows the system to better reflect real-world user behavior patterns observed in the data.

\subsection{Computational Considerations}

The extension to more complex POMDPs introduces significant computational challenges that fundamentally alter our solution approach. While our current model admits specialized solutions leveraging the structure of belief walks, this structure breaks down in the general POMDP setting. 

First, in our simplified model, we could effectively represent policies as sequences of actions by implicitly conditioning on the user remaining in the system. This representation was computationally advantageous as it reduced the policy search space. However, with multiple possible feedback types, we need to find a proper policy -- a mapping from beliefs to actions -- that specifies optimal recommendations for every possible belief state. This more complex policy structure is inherently harder to optimize than our current sequence-based approach.

This fundamental shift in policy representation directly impacts our branch-and-bound algorithm. In our current model, the algorithm exploits the fact that only positive feedback (``like'') is informative for future recommendations, as negative feedback ends the episode. However, when policies must account for multiple feedback types, every observation could potentially influence future recommendations. As a result, evaluating each policy prefix requires considering multiple possible feedback branches, significantly expanding the solution space that must be explored.

These challenges suggest two potential paths forward. One approach would be to analyze the specific dynamics induced by the extended model and develop specialized POMDP solvers that exploit any structural properties we can identify. The feasibility of this approach depends heavily on the particular model extensions chosen and the resulting state-action-observation structure. Alternatively, we could leverage existing general-purpose POMDP solvers, which are already quite efficient but might be further optimized for our specific recommendation domain by incorporating domain-specific heuristics or approximations.

\subsection{Theoretical Implications}

The extension to general POMDPs significantly impacts the theoretical analysis presented in this paper. The clean theoretical properties that we established, particularly in Section~\ref{sec:convergence-of-the-optimal-policy}, rely fundamentally on the deterministic nature of belief walks in our current model. When this determinism breaks down, several key theoretical challenges emerge:

\begin{itemize}
    \item Multiple possible future beliefs from each action prevent the definition of a unique belief walk
    \item Convergence results based on monotonic progression through the belief simplex no longer apply
\end{itemize}

New theoretical frameworks would be needed to analyze the behavior of optimal policies in these more general settings, likely focusing on probabilistic convergence properties rather than deterministic ones.

\subsection{Future Research Directions}

The extension to general POMDPs opens several promising research directions. First, developing specialized solvers that exploit the particular structure of recommendation systems could improve computational efficiency. Second, theoretical analysis of approximate solution methods could provide performance guarantees for practical implementations.

These extensions represent a natural evolution of our model toward more practical applications, though at the cost of some theoretical elegance. The challenge for future research lies in finding the right balance between model complexity and analytical tractability while maintaining computational feasibility.
}
\else{}
\fi

\end{document}